\documentclass{sadhana}
\usepackage{graphicx,amsmath,amssymb,verbatim,bm}
\allowdisplaybreaks

\DeclareFontFamily{U}{mathx}{\hyphenchar\font45}
\DeclareFontShape{U}{mathx}{m}{n}{
      <5> <6> <7> <8> <9> <10>
      <10.95> <12> <14.4> <17.28> <20.74> <24.88>
      mathx10
      }{}
\DeclareSymbolFont{mathx}{U}{mathx}{m}{n}
\DeclareMathSymbol{\bigplus}{1}{mathx}{"90}
\DeclareMathSymbol{\bigtimes}{1}{mathx}{"91}

\newcounter{lemqtypicalcq}
\newcounter{lemqtypicaldistance}
\newcounter{lemqtypicalcompleteness1}
\newcounter{lemqtypicalsoundness1}

\newcounter{thmqtypicalcq}
\newcounter{thmqtypicaldistance}
\newcounter{thmqtypicalcompleteness1}
\newcounter{thmqtypicalsoundness1}

\newtheorem{definition}{Definition}

\newtheorem{fact}{Fact}
\newtheorem{theorem}{Theorem}

\newenvironment{proof}{\noindent {\bf Proof:}}{\ \hfill \ $\Box$}

\newcommand{\vecl}{\mathbf{l}}
\newcommand{\vecx}{\mathbf{x}}

\newcommand{\cC}{\mathcal{C}}
\newcommand{\cE}{\mathcal{E}}
\newcommand{\cH}{\mathcal{H}}
\newcommand{\cL}{\mathcal{L}}
\newcommand{\cQ}{\mathcal{Q}}
\newcommand{\cU}{\mathcal{U}}
\newcommand{\cX}{\mathcal{X}}
\newcommand{\cY}{\mathcal{Y}}
\newcommand{\cZ}{\mathcal{Z}}

\newcommand{\hl}{\hat{l}}
\newcommand{\hQ}{\hat{Q}}
\newcommand{\hU}{\hat{U}}
\newcommand{\hX}{\hat{X}}
\newcommand{\hY}{\hat{Y}}
\newcommand{\hPi}{\hat{\Pi}}

\newcommand{\C}{\mathbb{C}}
\newcommand{\I}{\mathbb{I}}

\newcommand{\chan}{\mathfrak{C}}

\DeclareMathOperator*{\E}{{\rm {\bf E}}\,}
\DeclareMathOperator*{\Tr}{{\rm Tr}\;}

\newcommand{\zero}{\leavevmode\hbox{\small l\kern-3.5pt\normalsize0}}
\newcommand{\one}{\leavevmode\hbox{\small1\kern-3.8pt\normalsize1}}

\newcommand{\cupdot}{\mathbin{\mathaccent\cdot\cup}}

\newcommand{\elltwo}[1]{\left\|{ #1 }\right\|_2}
\newcommand{\ellone}[1]{\left\|{ #1 }\right\|_1}
\newcommand{\ellinfty}[1]{\left\|{ #1 }\right\|_\infty}

\newcommand{\ket}[1]{| #1 \rangle}
\newcommand{\bra}[1]{\langle #1 |}
\newcommand{\ketbra}[1]{\ket{#1}\bra{#1}}

\begin{document}

\title{{\bf Inner bounds via simultaneous decoding in quantum
network information theory
}}

\author{Pranab Sen\textsuperscript{1,*}}
\affilOne{\textsuperscript{1}
School of Technology and Computer Science, Tata Institute of Fundamental
Research, Mumbai 400005, India.
Email: {\sf pranab.sen.73@gmail.com}
}

\twocolumn[{

\maketitle

\begin{abstract}
We prove new inner bounds for several multiterminal channels with
classical inputs and quantum outputs. Our inner bounds are all proved
in the one-shot setting and are natural analogues of the best classical
inner bounds for the respective channels. For some of these channels,
similar quantum inner bounds were unknown even in the asymptotic 
independent and identically distributed setting. We prove 
our inner bounds by appealing to a new 
classical-quantum joint typicality lemma established in a
companion paper~\cite{sen:oneshot}. This lemma allows us to lift to
the quantum setting
many inner bound proofs for classical multiterminal channels that use
intersections and unions of typical sets.
\end{abstract}



}]



\markboth{Pranab Sen}{Inner bounds via simultaneous decoding in
quantum network information theory}

\section{Introduction}
An important technical tool used in proving inner bounds in 
classical network
information theory is the so-called conditional joint typicality lemma
\cite{book:elgamalkim}. 
What is equally important but often not emphasised are the implicit
{\em union} and {\em intersection} arguments used in the inner bound 
proofs. 
For quantum channels, proving a joint typicality lemma that can
withstand union and intersection arguments was a big bottleneck. As
a result of this bottleneck, many inner bounds in classical network
information theory were hitherto not known to be extendable to the 
quantum setting. 

Most inner bounds in information theory were first proved in the 
traditional setting of many
independent and identically distributed (iid) uses of a classical 
communication channel. Recently, 
attention has shifted 
to proving inner bounds in the {\em one-shot} setting where the 
classical or quantum channel can be
used only once. This is the most general setting. The aim is to prove
good one-shot inner bounds which ideally yield the best known inner
bounds when restricted to the asymptotic iid and asymptotic non-iid
(information spectrum) settings. In the one-shot setting, the importance
of union and intersection arguments increases and they often need to be 
made explicit. This is because the technique of 
time sharing often used in the asymptotic iid setting
does not apply in the one-shot setting. In other words, the one-shot
setting forces us to look for so-called {\em simultaneous decoders} for
multiterminal channels.
The inner bound analyses for simultaneous decoders generally use
union and intersection arguments.

Fawzi {\it et al}~\cite{fawzi:interference} and 
Sen~\cite{sen:interference} did construct a simultaneous decoder for
the two sender multiple access channel with classical inputs and quantum
output (cq-MAC) but their constructions, which 
were given in the 
asymptotic iid setting, are not known to work in the one-shot setting.
Qi, Wang and Wilde~\cite{qi:simultaneous} constructed a one-shot 
simultaneous
decoder for the cq-MAC with an arbitrary number of senders, but their
achievable rates restricted to the asymptotic iid setting are inferior
to the optimal rates obtained by Winter~\cite{winter:cqmac} using
successive cancellation.
Thus, for more than two senders a 
simultaneous decoder for the cq-MAC achieving optimal rates was
hitherto unknown even in the asymptotic
iid setting.  A simultaneous decoder for the 
MAC with three senders is used as a crucial ingredient in
the proof of the
Han-Kobayashi inner bound for the interference channel
\cite{HanKobayashi}, even in the asymptotic iid classical setting. 
Thus, the
lack of a simultaneous decoder for the asymptotic iid quantum setting
is a bottleneck, which was sidestepped by Sen~\cite{sen:interference} 
by constructing a simultaneous decoder for a restricted type of three
sender cq-MAC which sufficed to prove the Han-Kobayashi inner bound in
the asymptotic iid setting
for sending classical information over a quantum interference channel.
Hirche, Morgan and Wilde~\cite{hirche:interference} also 
proved the Han-Kobayashi inner bound for sending classical information
over a quantum interference channel in the asymptotic iid setting.
They did so using successive cancellation and polar coding.
However, both Sen's and Hirche {\it et al}'s techniques are tied to 
the asymptotic iid setting
and do not give any non-trivial inner
bound for the interference channel in the one-shot setting. 
Additionally, those techniques do not seem to give any 
non-trivial inner bound for
the entanglement assisted interference channel even in the asymptotic
iid quantum setting.

Very recently, in a companion paper Sen~\cite{sen:oneshot} proved 
a one-shot quantum joint typicality
lemma that possesses strong union and intersection properties.
Using that lemma, he also constructed a one-shot simultaneous decoder 
for the cq-MAC with an arbitrary number of senders. 
In this paper\footnote{Journal version of \cite{sen:simultaneous}},
we use the quantum joint typicality lemma from \cite{sen:oneshot} to 
obtain for the first
time non-trivial one-shot inner bounds for sending classical information
over several multiterminal quantum channels.
The channels that we consider are the broadcast channel and interference
channel, both without and with entanglement assistance.
For both 
channels our one-shot quantum inner bounds are the natural analogues
of the best known classical asymptotic iid inner bounds, and reduce
to them in the iid limit.

\subsection{Organisation of the paper}
In the next section, we state some preliminary facts which will be
useful throughout the paper. 
In Section~\ref{sec:cqtypical}, we state two simple versions of
Sen's quantum joint typicality lemma \cite{sen:oneshot} which suffice
for the applications in this paper.
In Section~\ref{sec:broadcast}, we prove a one-shot Marton inner bound 
with common message~\cite{Marton} for sending classical information through
unassisted as well as entanglement assisted quantum broadcast channel.
Section~\ref{sec:interference} proves the achievability of the
Han-Kobayashi~\cite{HanKobayashi} and 
Chong-Motani-Garg-El Gamal~\cite{CMGElGamal} inner bounds for 
one-shot use of a 
cq-interference channel.
Finally, we make some concluding remarks and list some open problems
in Section~\ref{sec:conclusions}.

\section{Preliminaries}
\label{sec:preliminaries}
All Hilbert spaces in this paper are finite dimensional. 
The symbol $\oplus$ always denotes the orthogonal direct
sum of Hilbert spaces.
For a subspace $X$ of a Hilbert space $\cH$, 
let $\Pi^{\cH}_X$ denote the orthogonal projection in $\cH$ 
onto $X$. When clear from the context, we may use $\Pi_X$ instead of
$\Pi^{\cH}_X$ for brevity of notation.

By a quantum 
state or a density matrix in a Hilbert space $\cH$, we mean a Hermitian,
positive semidefinite linear operator on $\cH$ with trace equal to one.
By a POVM
element $\Pi$ in $\cH$, we mean a Hermitian
positive semidefinite linear operator on $\cH$ with eigenvalues between
$0$ and $1$. Stated in terms of inequalities on Hermitian operators,
$\zero \leq \Pi \leq \one$, where $\zero$, $\one$ denote the zero and
identity operators on $\cH$. 

Let $\elltwo{v}$ denote the $\ell_2$-norm of a vector $v \in \cH$.
For an operator $A$ on $\cH$, we use $\ellone{A}$ to denote the
Schatten $\ell_1$-norm, also known as trace norm, of $A$, 
which is nothing but
the sum of singular values of $A$. We use $\ellinfty{A}$ to denote
the Schatten $\ell_\infty$-norm, also known as operator norm, 
of $A$, which is 
nothing but the largest singular value of $A$. For operators $A$, $B$
on $\cH$, we have the inequality 
\[
|\Tr [A B]| \leq \ellone{A B} \leq 
\min\{\ellone{A} \ellinfty{B}, \ellinfty{A} \ellone{B}\}.
\]

Let $\cX$ be a finite set. By a {\em classical-quantum} (hereafter called
cq for short) state on
$\cX \cH$ we mean a quantum state of the form
$
\rho^{\cX \cH} =
\sum_{x \in \cX} p_x \ketbra{x}^{\cX} \otimes \rho_x^{\cH},
$
where $x$ ranges over computational basis vectors of $\cX$ viewed
as a Hilbert space, $\{p_x\}_{x \in \cX}$ is a probability distribution
on $\cX$ and the operators $\rho_x$ for all $x \in \cX$ are 
quantum states in
$\cH$. We will also use the terminology that $\rho$ is classical on
$\cX$ and quantum on $\cH$. In this paper, superscripts in the notation
for a quantum state will denote the Hilbert space in which it lies. 
A similar convention will be used for classical probability distributions.

In this paper all our quantum operations will be trace non-increasing
completely positive superoperators,
generalising unitary evolution, POVM measurement and tracing out 
subsystems. For brevity, we will use the term {\em superoperator} 
to denote such operations. An expression like
$\chan^{A_1 A_2 \rightarrow B_1 B_2 B_3}$ will denote a superoperator
taking operators on $A_1 \otimes A_2$ to operators on 
$B_1 \otimes B_2 \otimes B_3$.
We use $\I^{A}$ to denote the identity superoperator on $A$.
When there is a need for very
precise notation, we will use expressions like 
$(\chan^{A_1 A_2 \rightarrow B}(\rho^{A_1 A_2}))^B$ to denote
the (possibly subnormalised) quantum state in $B$ obtained by
applying the superoperator $\chan$ to quantum state $\rho^{A_1 A_2}$.

For a positive integer $c$, we will use $[c]$ to denote the set 
$\{1, 2, \ldots, c\}$. If $c = 0$, we define $[c] := \{\}$.
We shall study systems that are classical on $\cX^{\otimes [c]}$ and
quantum on $\cH$. 
If $\vecx$ is a computational basis vector of $\cX^{\otimes [c]}$,
for a subset $S \subseteq [c]$,
$\vecx_S$ will denote its restriction to the system $\cX^{\otimes S}$. 
Thus, $\vecx \equiv \vecx_{[c]}$. We also use $\vecx_S$ to denote
computational basis vectors of $\cX^{\otimes S}$ without reference to
the systems in $[c] \setminus S$. 
The notation
$(\cdot)^{\otimes S}$ denotes a tensor product only for the coordinates
in $S$. We will use the notation $(S_1, \ldots, S_l) \subseteq [c]$ to
denote a collection of subsets $S_1, \ldots, S_l$, $l > 0$ of $[c]$. 
Note that
order does not matter in describing a collection of subsets of $[c]$.

We will need Winter's gentle measurement lemma~\cite{winter:gentle}.
\begin{fact}[\cite{OgawaNagaoka}]
\label{fact:gentle}
Let $\Lambda$ be a POVM element and $\rho$ be a quantum state such
that $\Tr[\Lambda \rho] \geq 1 - \epsilon$. Then,
$
\ellone{\rho - \Lambda^{1/2} \rho \Lambda^{1/2}} \leq 
2 \sqrt{\epsilon}.
$
\end{fact}

We recall the definition of the {\em hypothesis testing relative
entropy} given by Wang and Renner~\cite{wang:DepsH}. Very similar
quantities were defined and used in earlier 
works~\cite{buscemi:qchannel, brandao:entanglement}.
\begin{definition}
\label{def:hyptestingrelentropy}
Let $\alpha$, $\beta$ be two quantum states in the same Hilbert space.
Let $0 \leq \epsilon < 1$.
Then the {\em hypothesis testing relative entropy} of $\alpha$ with respect
to $\beta$ is defined by
\[
D^{\epsilon}_H(\alpha \| \beta) :=
\max_{\Pi: \Tr [\Pi \alpha] \geq 1 - \epsilon}
-\log \Tr [\Pi \beta],
\]
where the maximisation is over all POVM elements $\Pi$ acting on 
the Hilbert space.
\end{definition}
The definition quantifies the minimum probability of `accepting' $\beta$
by a POVM element $\Pi$ that `accepts' $\alpha$ with probability
at least $1 - \epsilon$.
From the definition, it is
easy to see that if $\epsilon < \epsilon'$,
$D^{\epsilon}_H(\alpha \| \beta) < D^{\epsilon'}_H(\alpha \| \beta)$.
We now define the {\em hypothesis testing mutual information} of 
a bipartite quantum state $\rho^{AB}$.
\begin{definition}
Let $0 \leq \epsilon < 1$.
Let $\rho^{AB}$ be a quantum state in a bipartite system $AB$. 
The {\em hypothesis testing mutual information} is defined as
$
I^\epsilon_H(A : B)_\rho := 
D^\epsilon_H(\rho^{AB} \| \rho^A \otimes \rho^B).
$
\end{definition}
For a cq-state, we can define the {\em hypothesis testing conditional
mutual information}.
\begin{definition}
Let $0 \leq \epsilon < 1$.
Let $\rho^{ABC}$ be a state which is classical on $A$ and quantum on
$BC$. It can be expressed as 
$
\rho^{ABC} =
\sum_a p(a) \ketbra{a}^A \otimes \rho_a^{BC}.
$
Consider a state $\sigma^{ABC}$ which is classical on $A$ and quantum on
$BC$ defined as
$
\sigma^{ABC} =
\sum_a p(a) \ketbra{a}^A \otimes \rho_a^{B} \otimes \rho_a^C.
$
The {\em hypothesis testing conditional mutual information} is defined as
$
I^\epsilon_H(B : C | A)_\rho := 
D^\epsilon_H(\rho^{ABC} \| \sigma^{ABC}).
$
\end{definition}

Let $0 \leq \epsilon \leq 1$. Let $P$, $Q$ be probablity distributions
on the same sample space $\cX$. For non-negative vectors $v_1, v_2$ 
supported on $\cX$, we use the notation $v_1 \leq v_2$ to denote
$v_1(x) \leq v_2(x)$ for all sample points $x \in \cX$. 
We now define the smooth max relative entropy of 
$P$ with respect to $Q$. The definition below is obtained by
taking the classical
version of the quantity defined by Datta~\cite{Datta}, coupled with
the observation that there exists a minimising $P'$ in the definition 
satisfying $P' \leq P$. This condition will be useful when we 
prove a one-shot mutual covering lemma in Fact~\ref{fact:mutualcovering}.
\begin{definition}
The $\epsilon$-smooth max relative entropy of $P$ with respect to $Q$
is defined as
\[
D^\epsilon_\infty(P \| Q) :=
\min_{0 \leq P' \leq P: \ellone{P - P'} \leq \epsilon}
\max_{x \in \cX} \log \frac{P'(x)}{Q(x)},
\] 
where $\frac{0}{0} := 1$. Note that $D^\epsilon_\infty(P \| Q)$ can
be $+\infty$ if the support of $P$ is not contained in the support of
$Q$ and $\epsilon$ is small.
\end{definition}
For completeness, we recall Datta's definition of $\epsilon$-smooth max 
relative entropy of quantum state $\rho$ with respect to quantum
state $\sigma$:
\[
D^\epsilon_\infty(\rho \| \sigma) :=
\min_{\rho': \ellone{\rho - \rho'} \leq \epsilon}
\;
\min_{\lambda \in \mathbb{R}: \rho \leq 2^\lambda \sigma} 
\;
\lambda.
\] 
Note that $D^\epsilon_\infty(\rho \| \sigma)$ can
be $+\infty$ if the support of $\rho$ is not contained in the support of
$\sigma$ and $\epsilon$ is small.

For a joint probability distribution $P$  on the sample space 
$\cX \times \cY \times \cZ$, we
define the smooth max conditional mutual information as follows.
\begin{definition}
\label{def:Imax}
The $\epsilon$-smooth max mutual information between random 
variables
$X$ and $Y$ conditioned on $Z$ under the joint distribution $P$ is 
defined as
\[
I^\epsilon_\infty(X : Y | Z)_P :=
D^\epsilon_\infty(P^{XYZ} \| P^Z \times (P^X|Z) \times (P^Y|Z)),
\]
where the superscripts denote the sample spaces of the 
respective probability distributions, and 
$P^Z \times (P^X|Z) \times (P^Y|Z)$ denote the probability distribution
on $\cX \times \cY \times \cZ$ obtained by first taking a sample
according to the marginal on $\cZ$ followed by independently taking a
pair of samples according to the marginals on $\cY$ and $\cZ$ conditioned
on the chosen sample from $\cZ$.
\end{definition}

We now define the so-called `restricted smooth conditional max mutual 
information' for a quantum state $\rho^{ZXY}$ which is classical on
$Z$ and quantum on $X$ and $Y$. Our definition is the conditional 
version of a quantity defined in \cite{anshu:broadcast}.
\begin{definition}
\label{def:Imaxcond}
The $\epsilon$-smooth restricted conditional max mutual information 
between $X$ and $Y$ conditioned on $Z$ under the state
$\rho^{XYZ}$, where $X$, $Y$ are quantum and $Z$ is classical,  
is defined as
\[
I^{\epsilon,\delta}_\infty(X : Y | Z)_\rho :=
D^{\epsilon,\delta}_\infty(\rho^{XYZ}
\| \rho^Z \otimes \rho^X|Z \otimes \rho^Y|Z),
\]
where the classical quantum state
$
\rho^Z \otimes \rho^X|Z \otimes \rho^Y|Z
$
is obtained by taking a sample $z$ according to the marginal probability 
distribution $\rho^Z$ followed by the tensor product of the marginal
quantum states $\rho^X|z$ and $\rho^Y|z$ obtained by conditioning on $z$,
and the smoothing in the definition of 
$D^{\epsilon,\delta}_\infty(\cdot \| \cdot)$ is done only over
classical quantum states $(\rho')^{XYZ}$ $\epsilon$-close to
$\rho^{XYZ}$ satistying $(\rho')^{XZ} \leq (1+\delta) \rho^{XZ}$
and $(\rho')^{YZ} \leq (1+\delta) \rho^{YZ}$.
\end{definition}

We next state a {\em one-shot mutual covering lemma} which strengthens
the one-shot mutual covering lemma of Radhakrishnan 
{\it et al}~\cite[Lemma 3]{radhakrishnan:broadcast}.
Our mutual covering lemma is closely related to the {\em bipartite 
convex split lemma} of Anshu, Jain and Warsi~\cite{anshu:broadcast} 
specialised to the classical setting. We state it in this form so that
it may be useful for other problems in network information theory.
For the broadcast channel, it allows us to give a clean one-shot
proof of Marton's inner bound with the added advantage of decoding 
Alice's `input random variables' exactly and not just `up to the band' as
in the traditional forms of Marton's inner bound. 
\begin{fact}[One-shot mutual covering lemma]
\label{fact:mutualcovering}
Let $(U_0, U_1, U_2)$ be a triple of random variables in the sample
space $\cU_0 \times \cU_1 \times \cU_2$ with joint 
distribution function $P^{U_0 U_1 U_2}$. 
Let $0 < \epsilon < 1$.
Define $I_\infty := I^\epsilon_\infty(U_1 : U_2 | U_0)_P$.
Let $r_1$, $r_2$ be positive integers such that
\[
r_1 + r_2 \geq I_\infty + 2 \log \frac{1}{\epsilon}.
\]

We now define two probability distributions on the set 
$
\cU_0 \times (\cU_1)^{2^{r_1}} \times (\cU_2)^{2^{r_2}} 
\times [2^{r_1}] \times [2^{r_2}]
$ 
as follows. 
\begin{enumerate}

\item
For the first distribution 
$(P_1)^{U_0 (U_1)^{2^{r_1}} (U_2)^{2^{r_2}} K_1 K_2}$,
define a new pair of random variables $(K_1, K_2)$ 
taking uniformly random values in $[2^{r_1}] \times [2^{r_2}]$. 
Choose first a sample $u_0$ according to the marginal $P^{U_0}$.
Choose independently a sample $(k_1, k_2)$ from $(K_1, K_2)$.
Let 
\begin{eqnarray*}
\lefteqn{\vec{U_1}^{-k_1} | u_0} \\ 
& := &
(U_1(1) \times \cdots \times 
U_1(k_1 - 1) \times U_1(k_1 + 1) \\
&    &
{} \times \cdots \times U_1(2^{r_1})
) | u_0
\end{eqnarray*}
be $(2^{r_1} - 1)$ independent copies of the random variable 
$U_1 | (U_0 = u_0)$.
Similarly, define 
\begin{eqnarray*}
\lefteqn{\vec{U_2}^{-k_2} | u_0} \\
& := & 
(U_2(1) \times \cdots \times 
U_2(k_2 - 1) \times U_2(k_2 + 1) \\ 
&    &
{} \times \cdots \times U_2(2^{r_2})
) | u_0
\end{eqnarray*}
to be $(2^{r_2} - 1)$ independent copies of the random variable 
$U_2 | (U_0 = u_0)$.
Let $(U_1(k_1), U_2(k_2)) | u_0$ denote
the distribution $P^{U_1 U_2} | (U_0 = u_0)$
on the $(k_1, k_2)$th copy. This completes the definition of
the probability distribution 
$(P_1)^{U_0 (U_1)^{2^{r_1}} (U_2)^{2^{r_2}} K_1 K_2}$ 
denoted in brief by
\begin{eqnarray*}
\lefteqn{(P_1)^{U_0 (U_1)^{2^{r_1}} (U_2)^{2^{r_2}} K_1 K_2}} \\ 
& := &
K_1 K_2 U_0 ((U_1(k_1), U_2(k_2)) | U_0) 
(\vec{U_1}^{-K_1} | U_0)
(\vec{U_2}^{-K_2} | U_0).
\end{eqnarray*}

\item
For the second distribution 
$(P_2)^{U_0 (U_1)^{2^{r_1}} (U_2)^{2^{r_2}} K_1 K_2}$,
choose first a sample $u_0$ according to the marginal $P^{U_0}$.
Let 
\[
\vec{U_1} | u_0 := (U_1(1) \times \cdots \times U_1(2^{r_1})) | u_0
\] 
be $2^{r_1}$ independent copies of the random variable 
$U_1 | (U_0 = u_0)$. 
conditioned on the sample from $U_0$.
Similarly, define 
\[
\vec{U_2} | u_0 := (U_2(1) \times \cdots \times U_2(2^{r_2})) | u_0
\] 
to be $2^{r_2}$ independent copies of the random variable 
$U_2 | (U_0 = u_0)$. 
A pair $(k_1, k_2) \in [2^{r_1}] \times [2^{r_2}]$ is now chosen
conditioned on the other random variables with exactly the same 
conditioning as in the distribution 
$(P_1)^{U_0 (U_1)^{2^{r_1}} (U_2)^{2^{r_2}} K_1 K_2}$.
We shall denote the complete distribution so obtained by
$(P_2)^{U_0 (U_1)^{2^{r_1}} (U_2)^{2^{r_2}} K_1 K_2}$ and denote 
it in brief by
\begin{eqnarray*}
\lefteqn{(P_2)^{U_0 (U_1)^{2^{r_1}} (U_2)^{2^{r_2}} K_1 K_2}} \\
& := &
U_0 (\vec{U_1} | U_0) (\vec{U_2} | U_0)
((K_1, K_2) | U_0 \vec{U_1} \vec{U_2}).
\end{eqnarray*}

\end{enumerate}

Then,
\[
\ellone{
(P_1)^{U_0 (U_1)^{2^{r_1}} (U_2)^{2^{r_2}} K_1 K_2} -
(P_2)^{U_0 (U_1)^{2^{r_1}} (U_2)^{2^{r_2}} K_1 K_2}
} \leq
4 \epsilon.
\]
\end{fact}
\begin{proof}
First, condition on a sample $u_0$ from the marginal $P^{U_0}$.
Consider now the distributions 
$(P_1)^{(U_1)^{2^{r_1}} (U_2)^{2^{r_2}}} | (U_0 = u_0)$,
$(P_2)^{(U_1)^{2^{r_1}} (U_2)^{2^{r_2}}} | (U_0 = u_0)$.
Suppose one can show that
\[
\ellone{
(P_1)^{(U_1)^{2^{r_1}} (U_2)^{2^{r_2}}} | (U_0 = u_0) -
(P_2)^{(U_1)^{2^{r_1}} (U_2)^{2^{r_2}}} | (U_0 = u_0) 
} \leq 4 \epsilon.
\]
This will imply that
\[
\ellone{
(P_1)^{U_0 (U_1)^{2^{r_1}} (U_2)^{2^{r_2}}} -
(P_2)^{U_0 (U_1)^{2^{r_1}} (U_2)^{2^{r_2}}}
} \leq 4 \epsilon.
\]
Now observe that the conditioning of $(K_1, K_2)$ on the other
random variables is exactly the same in the two distributions
$(P_1)^{U_0 (U_1)^{2^{r_1}} (U_2)^{2^{r_2}} K_1 K_2}$,
$(P_2)^{U_0 (U_1)^{2^{r_1}} (U_2)^{2^{r_2}} K_1 K_2}$.
This implies that
\begin{eqnarray*}
\lefteqn{
\ellone{
(P_1)^{U_0 (U_1)^{2^{r_1}} (U_2)^{2^{r_2}} K_1 K_2} -
(P_2)^{U_0 (U_1)^{2^{r_1}} (U_2)^{2^{r_2}} K_1 K_2}
} 
} \\
& = &
\ellone{
(P_1)^{U_0 (U_1)^{2^{r_1}} (U_2)^{2^{r_2}}} -
(P_2)^{U_0 (U_1)^{2^{r_1}} (U_2)^{2^{r_2}}}
} \leq 4 \epsilon.
\end{eqnarray*}

It only remains to show that
\[
\ellone{
(P_1)^{(U_1)^{2^{r_1}} (U_2)^{2^{r_2}}} | (U_0 = u_0) -
(P_2)^{(U_1)^{2^{r_1}} (U_2)^{2^{r_2}}} | (U_0 = u_0) 
} \\ 
\leq 
4 \epsilon.
\]
For this, apply the bipartite convex split lemma of 
Anshu, Jain and Warsi with the observation that 
for classical probability distributions the `smoothing' subdistribution 
$(P')^{U_0 U_1 U_2}$ of Definition~\ref{def:Imax} satisfies
$(P')^{U_0 U_1 U_2} \leq P^{U_0 U_1 U_2}$ which implies that $\delta = 0$ 
in Lemma~3 of \cite{anshu:broadcast}. The above inequality then
follows easily. 

This completes the proof of our one-shot mutual covering lemma.
\end{proof}

We shall use the so-called {\em pretty good measurement (PGM)}
\cite{belavkin:pgm1, belavkin:pgm2}, also known as
square root measurement, in order to
construct our decoders. Given a set of POVM elements $\Pi_m$, $m \in [M]$,
the pretty good measurement is a POVM defined as follows:
\[
\Lambda_m := 
\left(\sum_{m'} \Pi_{m'}\right)^{-1/2}
\Pi_m
\left(\sum_{m'} \Pi_{m'}\right)^{-1/2}
\]
We will use the famous Hayashi-Nagaoka~\cite{HayashiNagaoka} 
operator inequality
in order to analyse the decoding error of the PGM POVM.
\begin{fact}
\label{fact:HN}
\[
\one - \Lambda_m \leq 
2 (\one - \Pi_m) +
4 \sum_{m': m' \neq m} \Pi_{m'}.
\]
\end{fact}

\section{The quantum joint typicality lemma}
\label{sec:cqtypical}
We now state the versions of the
classical-quantum joint typicality lemma from 
\cite{sen:oneshot} which suffice
for the applications in this paper. 
\begin{fact}[cq joint typ. lem., intersec. case]
\label{fact:cqtypical}
Let $\cH$, $\cL$ be two Hilbert spaces and $\cX$ be a finite set. 
We will also use $\cX$ to denote the Hilbert space with computational basis
elements indexed by the set $\cX$. Let $c$ be a non-negative integer.
Let $A$ denote a quantum register with Hilbert space $\cH$.
For every $\vecx \in \cX^c$, 
let $\rho_\vecx$ be a quantum state in $A$.
Consider the extended quantum system 
\[
A' := 
(\cH \otimes \C^2) \oplus 
\bigoplus_{S: \{\} \neq S \subseteq [c]}
(\cH \otimes \C^2) \otimes \cL^{\otimes |S|}.
\]
Also define the {\em augmented} classical system 
$\cX' := \cX \otimes \cL$.

Below, $\vecx$, $\vecl$ denote computational basis vectors of
$\cX^{[c]}$, $\cL^{\otimes [c]}$.
Let $p(\cdot)$ be a probability distribution on the vectors $\vecx$.
Define the classical quantum state
\[
\rho^{\cX_[c] A} :=
\sum_\vecx 
p(\vecx) \ketbra{\vecx}^{\cX_{[c]}} \otimes
\rho_\vecx^{A}.
\]
Let $\frac{\one^{\cL^{\otimes c}}}{|\cL|^{c}}$ denote the completely
mixed state on $c$ tensor copies of $\cL$. 
View 
$
\rho_\vecx^{A} 
\otimes (\ketbra{0})^{\C^2}
$
as a state in $A'$ under the natural embedding viz. the
embedding is into the first summand of $A'$
defined above. Similarly, view 
$
\rho^{\cX_{[c]} A} 
\otimes (\ketbra{0})^{(\C^2)}
\otimes \frac{\one^{\cL^{\otimes c}}}{|\cL|^{c}} 
$
as a state in $\cX'_{[c]} A'$ under the natural 
embedding.

Let $0 \leq \epsilon, \delta \leq 1$. 
Let $(S_1, S_2, S_3)$ be disjoint subsets of $[c]$ such that
$S_1 \cup S_2 \cup S_3 = [c]$.
We allow $S_1$ or $S_3$ or both to be empty, and denote the triple by 
$(S_1, S_2, S_3) \dashv [c]$. 
Choose $\cL$ to have dimension 
$|\cL| = \frac{3^{13} |\cH|^4}{2^4 (1 - \epsilon)^6}$.
Then, there is a state $\rho'$ and a POVM element $\Pi'$ in 
$\cX'_{[c]} A'$ such that:
\begin{enumerate}
\item
\setcounter{lemqtypicalcq}{\value{enumi}}
The state $\rho'$ and POVM element $\Pi'$ are classical 
on $\cX^{\otimes [c]} \otimes \cL^{[c]}$ and quantum on 
$A'$. 
More precisely, $\rho'$, $\Pi'$ can be expressed as
\begin{eqnarray*}
(\rho')^{\cX'_{[c]} A'} 
& = &
|\cL|^{-c} 
\sum_{\vecx, \vecl}
p(\vecx) 
\ketbra{\vecx}^{\cX_{[c]}} \otimes
\ketbra{\vecl}^{\cL_{[c]}} \otimes
(\rho')_{\vecx, \vecl, \delta}^{A'}, \\
(\Pi')^{\cX'_{[c]} A'} 
& = &
\sum_{\vecx, \vecl}
\ketbra{\vecx}^{\cX_{[c]}} \otimes
\ketbra{\vecl}^{\cL_{[c]}} \otimes
(\Pi')_{\vecx,\vecl,\delta}^{A'}, \\
\end{eqnarray*}
where 
$(\rho')_{\vecx,\vecl,\delta}^{A'}$, 
$(\Pi')_{\vecx,\vecl,\delta}^{A'}$ are 
quantum states and POVM elements respectively
for all computational basis vectors 
$\vecx \in \cX^{\otimes [c]}$,
$\vecl \in \cL^{\otimes [c]}$;

\item
\setcounter{lemqtypicaldistance}{\value{enumi}}
\[
\ellone{
(\rho')^{\cX'_{[c]} A'} - 
\rho^{\cX_{[c]} A} 
\otimes (\ketbra{0})^{\C^2}
\otimes \frac{\one^{\cL^{\otimes c}}}{|\cL|^c} 
} \leq
2^{\frac{c+1}{2} +1} \delta;
\]

\item
\setcounter{lemqtypicalcompleteness1}{\value{enumi}}
\[
\Tr [(\Pi')^{\cX'_{[c]} A'} (\rho')^{\cX'_{[c]} A'}] \geq
1 - 
\delta^{-2} 2^{2^{c+5}} 3^c \epsilon -
2^{\frac{c+1}{2}+1} \delta;
\]

\item
\setcounter{lemqtypicalsoundness1}{\value{enumi}}
Let $S \subseteq [c]$.
Let $\vecx_S$, $\vecl_S$ be computational basis vectors in 
$\cX^{\otimes S}$, $\cL^{\otimes S}$.
In the following definition, let $\vecx'_{\bar{S}}$,
$\vecl'_{\bar{S}}$ range over all
computational basis vectors of $\cX^{\otimes ([c] \setminus S)}$,
$\cL^{\otimes ([c] \setminus S)}$.
Define a state in $A'$,
\[
(\rho')_{\vecx_S, \vecl_S, \delta}^{A'} := 
|\cL|^{-|\bar{S}|} 
\sum_{\vecx'_{\bar{S}}, \vecl'_{\bar{S}}} 
p(\vecx'_{\bar{S}} | \vecx_S)
(\rho')_{\vecx_S \vecx'_{\bar{S}}, 
         \vecl_S \vecl'_{\bar{S}}, \delta
        }^{A'}.
\]
Analogously define 
\[
\rho_{\vecx_S}^A :=
\sum_{\vecx'_{\bar{S}}}
p(\vecx'_{\bar{S}} | \vecx_S)
\rho_{\vecx_S \vecx'_{\bar{S}}}^A.
\]
Let $(S_1, S_2, S_3) \dashv [c]$.
Define 
\begin{eqnarray*}
\lefteqn{(\rho')_{(S_1, S_2, S_3)}^{\cX'_{[c]} A'} }\\
& := &
|\cL|^{-c} 
\sum_{\vecx_{S_1}}
p(\vecx_{S_1}) 
\ketbra{\vecx_{S_1}}^{\cX_{S_1}} \otimes
\ketbra{\vecl_{S_1}}^{\cL_{S_1}} \\
&    &
~~~~~~~~~~~
{} \otimes
\left(
\sum_{\vecx_{S_2}}
p(\vecx_{S_2} | \vecx_{S_1})
\ketbra{\vecx_{S_2}}^{\cX_{S_2}}
\right. \\
&   &
~~~~~~~~~~~~~~~~~~
\left.
{} \otimes
\ketbra{\vecl_{S_2}}^{\cL_{S_2}} 
\right) \\
&    &
~~~~~~~~~~~
{} \otimes
\left(
\sum_{\vecx_{S_3}}
p(\vecx_{S_3} | \vecx_{S_1})
\ketbra{\vecx_{S_3}}^{\cX_{S_3}} 
\right. \\
&   &
~~~~~~~~~~~~~~~~~~
\left.
{} \otimes
\ketbra{\vecl_{S_3}}^{\cL_{S_3}} \otimes 
(\rho')_{\vecx_{S_1 \cup S_3}, \vecl_{S_1 \cup S_3}, \delta}^{A'}
\right), \\
\lefteqn{\rho_{(S_1, S_2, S_3)}^{\cX_{[c]} A} } \\
& := &
\sum_{\vecx_{S_1}}
p(\vecx_{S_1}) 
\ketbra{\vecx_{S_1}}^{\cX_{S_1}} \\
&    &
~~~~~~~~~~~
{} \otimes
\left(
\sum_{\vecx_{S_2}}
p(\vecx_{S_2} | \vecx_{S_1})
\ketbra{\vecx_{S_2}}^{\cX_{S_2}} 
\right) \\
&    &
~~~~~~~~~~~
{} \otimes
\left(
\sum_{\vecx_{S_3}}
p(\vecx_{S_3} | \vecx_{S_1})
\ketbra{\vecx_{S_3}}^{\cX_{S_3}} \otimes
\rho_{\vecx_{S_1 \cup S_3}}^{A}
\right).
\end{eqnarray*}

Then,
\[
\Tr [
(\Pi')^{\cX'_{[c]} A'} 
(\rho')_{(S_1, S_2, S_3)}^{\cX'_{[c]} A'}
]
\leq
2^{-I_H^{\epsilon}(X_{S_2} : A X_{S_3} | X_{S_1})_\rho},
\]
where 
$
I_H^{\epsilon}(X_{S_2} : A X_{S_3} | X_{S_1})_\rho :=
D_H^{\epsilon}
    (
      \rho^{\cX_{[c]} A} \| 
      \rho^{\cX_{[c]} A}_{(S_1, S_2, S_3)}
    ).
$
\end{enumerate}
\end{fact}
Informally speaking, the above lemma guarantees the existence of
a {\bf single} POVM element $\Pi'$ with robust properties that 
serves as an 
`intersection' of the individual POVM elements
achieving the hypothesis testing relative entropy quantities
arising from the state $\rho^{\cX_{[c]} A}$ by considering all possible
collections of subsets of $[c]$.

We next state a more general classical quantum joint typicality lemma
that guarantees the existence of
a {\bf single} POVM element $\Pi'$ with robust properties that serves as a 
`union of intersection' of individual POVM elements.
\begin{figure*}
\begin{center}
\includegraphics[width=\textwidth]{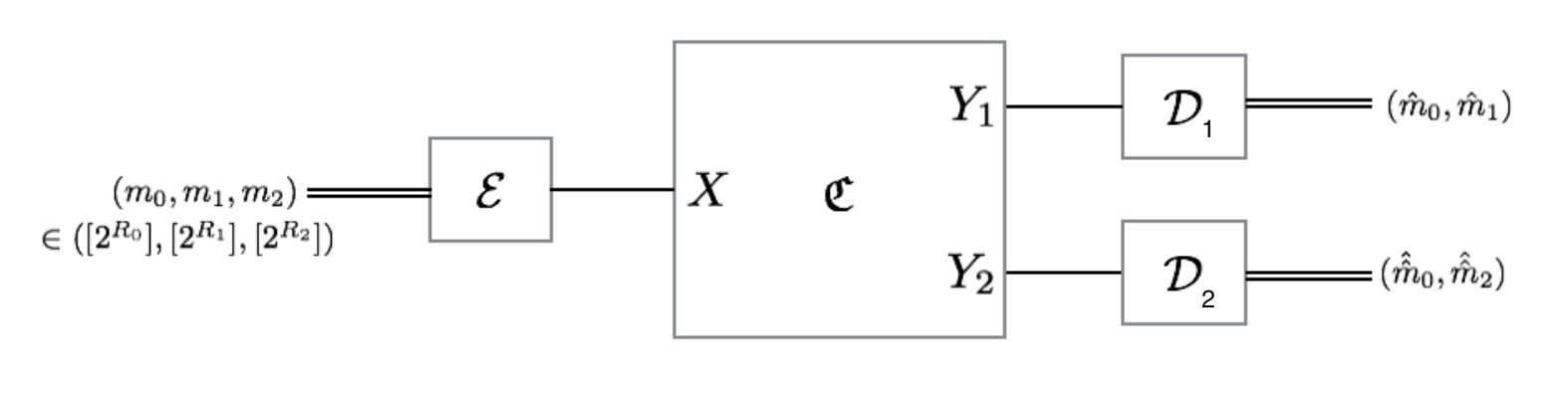}
\end{center}
\caption{Quantum broadcast channel without entanglement assistance.}
\label{fig:broadcast}
\end{figure*}
\begin{fact}[cq joint typ. lem., gen. case]
\label{fact:gencqtypical}
Let $\cH$, $\cL$ be Hilbert spaces and $\cX$ be a finite set. We will
also use $\cX$ to denote the Hilbert space with computational basis
elements indexed by the set $\cX$. Let $c$ be a non-negative integer.
Let  $A$ denote a quantum register with Hilbert space $\cH$.
For every $\vecx \in \cX^c$, 
let $\rho_\vecx$ be a quantum state in $A$.
Let $t$ be a positive integer.
Let $\vecx^{t}$ denote a $t$-tuple of elements of $\cX^c$;
we shall denote its $i$th element by $\vecx^{t}(i)$.
Consider the 
extended quantum system $\hat{A}$ where
$\hat{A} \cong A' \otimes \C^2 \otimes \C^{t+1}$, and
$A'$ is defined as
\[
A' := 
(\cH \otimes \C^2) \oplus 
\bigoplus_{S: i \in S \subseteq [c] \cupdot [k]}
(\cH \otimes \C^2) \otimes \cL^{\otimes |S|}.
\]
Also define the {\em augmented} classical system 
$\hat{\cX} := \cX \otimes \cL$.

Below, $\vecx$, $\vecl$ denote computational basis vectors of
$\cX^{[c]}$, $\cL^{\otimes [c]}$.
Let $p(\cdot)$ denote a probability distribution on 
the vectors $\vecx$. Let $p(1; \cdot), \ldots, p(t; \cdot)$ denote
probability distributions on $\vecx^t$ such that the marginal
of $p(i; \vecx^t)$ on the $i$th element is $p(\vecx^t(i))$.
For $i \in [t]$, define the classical quantum states
\[
\rho^{(\cX_[c])^t A}(i) :=
\sum_{\vecx^{t}} 
p(i; \vecx^t) \ketbra{\vecx^t}^{(\cX_{[c]})^t} \otimes
\rho_{\vecx^t(i)}^{A}.
\]
Let $\frac{\one^{\cL^{\otimes c}}}{|\cL|^{c}}$ denote the completely
mixed state on $c$ tensor copies of $\cL$. 
View 
$
\rho_\vecx^{A} 
\otimes (\ketbra{0})^{\C^2}
\otimes (\ketbra{0})^{\C^2}
\otimes (\ketbra{0})^{\C^{t+1}}
$
as a state in $\hat{A}$ under the natural embedding viz. the
embedding is into the first summand of $A'$
defined above tensored with $\C^2 \otimes \C^{t+1}$. 
Similarly, view 
$
\rho^{(\cX_{[c]})^t A}(i)
\otimes (\ketbra{0})^{\C^2}
\otimes \frac{\one^{\cL^{\otimes ct}}}{|\cL|^{ct}} 
\otimes (\ketbra{0})^{\C^2}
\otimes (\ketbra{0})^{\C^{t+1}}
$
as a state in $(\hat{\cX}_{[c]})^{\otimes t} \hat{A}$ under the natural 
embedding.

Let $0 \leq \alpha, \epsilon, \delta \leq 1$. 
Choose $\cL$ to have dimension 
$|\cL| = \frac{3^{13} |\cH|^4}{2^4 (1 - \epsilon)^6}$.
Then, there are states $\rho'(1), \ldots, \rho'(t)$ and a  
POVM element $\hat{\Pi}$ in 
$(\hat{\cX}_{[c]})^{\otimes t} \hat{A}$ such that:
\begin{enumerate}
\item
\setcounter{thmqtypicalcq}{\value{enumi}}
The states $\rho'(1), \ldots, \rho'(t)$ and 
POVM element $\hat{\Pi}$ are classical 
on $\cX^{\otimes [ct]} \otimes \cL^{[ct]}$ and quantum on 
$\hat{A}$.
More precisely, $\rho'(i)$, $i \in [t]$, 
$\hat{\Pi}$ can be expressed as
\begin{eqnarray*}
\lefteqn{(\rho'(i))^{(\hat{\cX}_{[c]})^t \hat{A}} } \\
& = &
|\cL|^{-ct} 
\sum_{\vecx^t, \vecl^t}
p(i; \vecx^t) 
\ketbra{\vecx^t}^{(\cX_{[c]})^{\otimes t}} \otimes
\ketbra{\vecl^t}^{(\cL_{[c]})^{\otimes t}} \\
&  &
~~~~~
{} \otimes
(\rho')_{\vecx^t(i), \vecl^t(i), \delta}^{A'}
\otimes (\ketbra{0})^{\C^2}
\otimes (\ketbra{0})^{\C^{t+1}}, \\
\lefteqn{(\hPi)^{(\hat{\cX}_{[c]})^t \hat{A}} } \\
& = &
\sum_{\vecx^t, \vecl^t}
\ketbra{\vecx^t}^{(\cX_{[c]})^{\otimes t}} \otimes
\ketbra{\vecl^t}^{(\cL_{[c]})^{\otimes t}} \otimes
(\hat{\Pi})_{\vecx^t,\vecl^t,\delta}^{\hat{A}}, \\
\end{eqnarray*}
where 
$(\rho')_{\vecx,\vecl,\delta}^{A'}$ 
are quantum states for all computational basis vectors 
$\vecx \in \cX^{\otimes [c]}$, $\vecl \in \cL^{\otimes [c]}$ and
$
(\hPi)_{\vecx^t,\vecl^t,\delta}^{\hat{A}}
$ 
are POVM elements for all computational basis vectors 
$\vecx^t \in \cX^{\otimes [ct]}$,
$\vecl^t \in \cL^{\otimes [ct]}$;

\item
\setcounter{thmqtypicaldistance}{\value{enumi}}
For all $i \in [t]$,
\[
\begin{array}{l}
\left\|
(\rho'(i))^{(\hat{\cX}_{[c]})^t \hat{A}} 
\right. \\
~~~
{} - 
(\rho(i))^{(\cX_{[c]})^t A} 
\otimes (\ketbra{0})^{\C^2}
\otimes \frac{\one^{\cL^{\otimes ct}}}{|\cL|^{ct}} \\
~~~~~~~~~~~~~
\left.
\otimes (\ketbra{0})^{\C^2}
\otimes (\ketbra{0})^{\C^{t+1}}
\right\|_1
\; \leq \;
2^{\frac{c+1}{2} +1} \delta;
\end{array}
\]

\item
\setcounter{thmqtypicalcompleteness1}{\value{enumi}}
For all $i \in [t]$,
\[
\Tr [
(\hat{\Pi})^{(\hat{\cX}_{[c]})^t \hat{A}} 
(\rho'(i))^{(\hat{\cX}_{[c]})^t \hat{A}}
] \geq
1 - 
\delta^{-2} 2^{2^{c+5}} 3^c \epsilon -
2^{\frac{c+1}{2}+1} \delta - \alpha;
\]

\item
\setcounter{thmqtypicalsoundness1}{\value{enumi}}
Let $S \subseteq [c]$.
Let $\vecx_{S}$, $\vecl_S$ be computational basis vectors in 
$\cX^{\otimes S}$, $\cL^{\otimes S}$.
In the following definition, let $\vecx'_{\bar{S}}$,
$\vecl'_{\bar{S}}$ range over all
computational basis vectors of $\cX^{\otimes ([c] \setminus S)}$,
$\cL^{\otimes ([c] \setminus S)}$.
Define states in $A'$,
\[
(\rho')_{\vecx_{S}, \vecl_{S}, \delta}^{A'} := 
|\cL|^{-|\bar{S}|} 
\sum_{\vecx'_{\bar{S}}, \vecl'_{\bar{S}}} 
p(\vecx'_{\bar{S}} | \vecx_S)
(\rho')_{\vecx_{S} \vecx'_{\bar{S}}, 
         \vecl_{S} \vecl'_{\bar{S}}, \delta
        }^{A'}.
\]
Analogously define 
\[
\rho_{\vecx_S}^{A} :=
\sum_{\vecx'_{\bar{S}}}
p(\vecx'_{\bar{S}} | \vecx_S)
\rho_{\vecx_{S} \vecx'_{\bar{S}}}^{A}.
\]

For $i \in [t]$, $S \subseteq [c]$, let $q_{i; S}(\cdot)$ be a
probability distribution on $\vecx^t$. 
Define 
\begin{eqnarray*}
\lefteqn{(\rho')_{i; S}^{(\hat{\cX}_{[c]})^t \hat{A}}} \\
& := &
|\cL|^{-ct} 
\sum_{\vecx^t}
q_{i; S}(\vecx^t)
\ketbra{\vecx^t}^{\cX^{\otimes [ct]}} \otimes
\ketbra{\vecl^t}^{\cL^{\otimes [ct]}} \\
&    &
~~~~~~~~~~~~~~
{} \otimes
(\rho')_{\vecx^t(i)_{S}, \vecl^t(i)_{S}, \delta}^{A'}, \\
\lefteqn{\rho_{i; S}^{(\cX_{[c]})^t A}} \\
& := &
\sum_{\vecx^t}
q_{i; S}(\vecx^t)
\ketbra{\vecx^t}^{\cX^{\otimes [ct]}} \otimes
\rho_{\vecx^t(i)_{S}}^{A}.
\end{eqnarray*}

Then,
\begin{eqnarray*}
\lefteqn{
\Tr [
(\hat{\Pi})^{(\hat{\cX}_{[c]})^t \hat{A}} 
(\rho')_{i; S}^{(\hat{\cX}_{[c]})^t \hat{A}}
]
} \\
& \leq &
\frac{1-\alpha}{\alpha}
\sum_{j=1}^t
2^{
    -D_H^{\epsilon}
     (
      \rho(j)^{(\cX_{[c]})^t A} \| 
      \rho_{i; S}^{(\cX_{[c]})^t A}
     )
  }.
\end{eqnarray*}
\end{enumerate}
\end{fact}

\section{Broadcast channel}
\label{sec:broadcast}
We now prove a one-shot Marton inner bound with 
common message for 
sending classical information through a quantum broadcast channel. Such
a result was not known earlier for a quantum broadcast channel even
in the asymptotic iid setting. The
analogous inner bound in the one-shot classical setting was proved by
Radhakrishnan, Sen and Warsi~\cite{radhakrishnan:broadcast} 
(see also Liu {\it et al}~\cite{liu:broadcast}). 
Radhakrishnan, Sen and Warsi
also proved Marton's inner bound, but without common message, in 
the one shot quantum setting. The version with common message
subsumes the version without, as well as the {\em superposition
coding technique} for a broadcast channel
\cite{cover:superposition, bergmans:superposition, 
savov:superposition}.
This problem was also studied earlier by Hirche and 
Morgan~\cite{hirche:broadcast}
for a two user binary input classical quantum broadcast channel.
Recently, Anshu, Jain and Warsi~\cite{anshu:qmac} proved
nearly matching one-shot inner and outer bounds for the quantum
broadcast channel without common message. However, their bounds  are 
not known to reduce to the standard
Marton bounds in the asymptotic iid limit.

In the problem of sending classical information with common message
through a quantum broadcast channel (q-BC),
the sender Alice has three classical messages $m_0 \in [2^{R_0}]$,
$m_1 \in [2^{R_1}]$, $m_2 \in [2^{R_2}]$, and she wants to send
$(m_0, m_1)$ to Bob and $(m_0, m_2)$ to Charlie. The parties have 
at their disposal
a quantum channel $\chan: X \rightarrow Y_1 Y_2$ with input 
Hilbert space $\cX$ and 
output Hilbert spaces $\cY_1$, $\cY_2$. Alice encodes $(m_0, m_1, m_2)$
into a quantum state $\sigma_{m_0, m_1, m_2}^X \in \cX$ and inputs it 
to $\chan$. The channel $\chan$ applies a superoperator to 
$\sigma_{m_0, m_1, m_2}$ and outputs a
quantum  state 
$
\rho_{m_0, m_1, m_2}^{Y_1 Y_2} := 
(\chan^{X \rightarrow Y_1 Y_2}(\sigma_{m_0, m_1, m_2}^X))^{Y_1 Y_2}$ 
jointly
supported in the Hilbert space $\cY_1 \otimes \cY_2$. Bob and Charlie
apply their respective decoding superoperators independently on
$\rho_{m_0, m_1, m_2}^{Y Z}$ in order to produce their respective guesses 
$(\hat{m}_0, \hat{m}_1)$, $(\hat{\hat{m}}_0, \hat{\hat{m}}_2)$ of the 
messages $m_0$, $m_1$, $m_2$. See Figure~\ref{fig:broadcast}.
Let $0 \leq \epsilon \leq 1$. Consider the uniform probability distribution
over the message sets.
We want that 
\[
\Pr[(\hat{m_0}, \hat{m}_1, \hat{\hat{m}}_0, \hat{\hat{m}}_2) \neq 
    (m_0, m_1, m_0, m_2)] \leq \epsilon,
\]
where the probability is over the choice of the messages and actions of the
encoder, channel and decoders. If there exists such encoding and
decoding schemes for a particular channel $\chan$, we say there exists
an $(R_0, R_1, R_2, \epsilon)$-quantum broadcast channel code for 
sending classical information through $\chan$.

It is possible to extend the classical proof of the one-shot Marton's inner
bound with common message of Radhakrishnan, Sen
and Warsi~\cite{radhakrishnan:broadcast} to the quantum setting by
using Fact~\ref{fact:cqtypical} to obtain the quantum analogues of
intersection operations used to define the sets 
$\mathcal{A}_{13}$, $\mathcal{A}_{24}$ just before Equation~(42) of their
paper. In this paper however, we give a different proof following
the style of Anshu, Jain and Warsi~\cite{anshu:broadcast} which 
we believe to be more transparent and intuitive.

We now state our one-shot Marton's inner bound with common message for 
transmitting classical information over a quantum broadcast channel.
\begin{theorem}[One-shot Marton, common message]
\label{thm:marton}
Let $\chan$ be a quantum broadcast channel. Let $\cU_0, \cU_1$, $\cU_2$ be
three new sample spaces and $(U_0, U_1, U_2)$ be a jointly distributed 
random variable on the sample space $\cU_0 \times \cU_1 \times \cU_2$. 
For every
element $(u_0, u_1, u_2) \in \cU_1 \times \cU_2$, let 
$\sigma_{u_0, u_1, u_2}^X$
be a quantum state in the input Hilbert space $\cX$ of $\chan$.
Consider the classical quantum state
\begin{eqnarray*}
\lefteqn{\rho^{U_0 U_1 U_2 Y_1 Y_2}} \\
& := &
\sum_{(u_0, u_1, u_2) \in \cU_0 \times \cU_1 \times \cU_2}
p_{U_0 U_1 U_2}(u_0, u_1, u_2) \\
&    &
~~~~~~~~~~~~~~~~~~~
\ketbra{u_0, u_1, u_2}^{U_0 U_1 U_2} \otimes
\chan(\sigma_{u_0, u_1, u_2}^X)^{Y_1 Y_2}.
\end{eqnarray*}
Let $R_0$, $R_1$, $R_2$, $\epsilon$, 
be such that
\begin{eqnarray*}
R_0 + R_1 
& \leq &
I_H^{\epsilon}(U_0 U_1 : Y_1) - 2 -  \log \frac{1}{\epsilon} \\
R_0 + R_2 
& \leq &
I_H^{\epsilon}(U_0 U_2 : Y_2) - 2 - \log \frac{1}{\epsilon} \\
R_0 + R_1 + R_2 
& \leq &
I_H^{\epsilon}(U_0 U_2 : Y_2) +
I_H^{\epsilon}(U_1 : Y_1 | U_0) \\
&      &
{} - 
I_\infty^{\epsilon}(U_1 : U_2 | U_0) - 4 - 4 \log \frac{1}{\epsilon} \\
R_0 + R_1 + R_2 
& \leq &
I_H^{\epsilon}(U_0 U_1 : Y_1) +
I_H^{\epsilon}(U_2 : Y_2 | U_0) \\
&      &
{} - 
I_\infty^{\epsilon}(U_1 : U_2 | U_0) - 4 - 4 \log \frac{1}{\epsilon} \\
2 R_0 + R_1 + R_2 
& \leq &
I_H^{\epsilon}(U_0 U_1 : Y_1) +
I_H^{\epsilon}(U_0 U_2 : Y_2) \\
&      &
{} - 
I_\infty^{\epsilon}(U_1 : U_2 | U_0) - 4 - 4 \log \frac{1}{\epsilon},
\end{eqnarray*}
where the mutual
information quantitites above are computed with respect to the
cq-state $\rho^{U_0 U_1 U_2 Y_1 Y_2}$. 
Then there exists an 
$(R_0, R_1, R_2, 2^7 \epsilon^{1/6})$-quantum broadcast channel code 
for sending
classical information through $\chan$. 
\end{theorem}
\begin{proof}
We follow the structure of Marton's common message inner bound proof 
as in Radhakrishnan {\it et al}~\cite{radhakrishnan:broadcast}
with the difference that we use the one-shot mutual covering lemma 
of Fact~\ref{fact:mutualcovering}
instead. 
Let $R_0$, $R_1$, $R_2$, $r_1$, $r_2$, $\epsilon$, 
be such that
\begin{eqnarray*}
R_0 + R_1 + r_1
& \leq &
I_H^{\epsilon}(U_0 U_1 : Y_1) - 2 - \log \frac{1}{\epsilon} \\
R_0 + R_2 + r_2
& \leq &
I_H^{\epsilon}(U_0 U_2 : Y_2) - 2 - \log \frac{1}{\epsilon} \\
R_1 + r_1
& \leq &
I_H^{\epsilon}(U_1 : Y_1 | U_0) - 2 - \log \frac{1}{\epsilon} \\
R_2 + r_2
& \leq &
I_H^{\epsilon}(U_2 : Y_2 | U_0) - 2 - \log \frac{1}{\epsilon} \\
r_1 + r_2
& = &
I_\infty^{\epsilon}(U_1 : U_2 | U_0) + 
2 \log \frac{1}{\epsilon},
\end{eqnarray*}
where the mutual
information quantitites above are computed with respect to the
cq-state $\rho^{U_0 U_1 U_2 Y_1 Y_2}$. 
Suppose we show that there is a
$(R_0, R_1, R_2, 2^7 \epsilon^{1/6})$-quantum broadcast channel 
code for sending classical information through $\chan$. 
Standard Fourier-Motzkin elimination
can be used to get rid of $r_1$ and $r_2$ and obtain the inner bound
in the statement of the theorem.

\medskip

\noindent
{\bf Codebook:}

\noindent
The codebook $\cC$ has $2^{R_0}$ {\em pages}. Each page consists of a 
two dimensional array of
`symbols' arranged in $2^{R_1 + r_1}$ rows and $2^{R_2 + r_2}$ columns.
We will index `entries' of $\cC$ by the $4$-tuple 
$(m_0, m_1, k_1, m_2, k_2)$ where
$m_i \in 2^{R_i}$, $k_j \in 2^{r_j}$. The codebook is generated randomly
as follows. First sample $u_0(1), \ldots, u_0(2^{R_0})$ 
independently according
to $p_{U_0}$. We will associate $u_0(m_0)$ with the $m_0$th page of $\cC$. 
Now, to generate the contents of the $m_0$th page, sample 
\[
u_1(m_0, 1), \ldots, u_1(m_0, 2^{R_1 + r_1}),
u_2(m_0, 1), \ldots, u_2(m_0, 2^{R_2 + r_2})
\]
independently according to
$p_{U_1 | u_0(m_0)}$, $p_{U_2 | u_0(m_0)}$. The codebook entry
$\cC(m_0, m_1, k_1, m_2, k_2)$ is the triple
\[
(u_0(m_0), u_1(m_0, m_1, k_1), u_2(m_0, m_2, k_2),
\]
where $(m_i, k_i)$ can be thought of as an element in $[2^{R_i + r_i}]$.
For $(m_0, m_1) \in [2^{R_0}] \times [2^{R_1}]$, define the 
`row band' $\cC(m_0, m_1)$ of samples
$
u_1(m_0, (m_1 - 1) 2^{r_1} + 1), \ldots,
u_1(m_0, m_1 2^{r_1}).
$
Similarly, one can define the `column band' $\cC(m_0, m_2)$ for each
$(m_0, m_2) \in [2^{R_0}] \times [2^{R_2}]$. 
For a triple $(m_0, m_1, m_2)$, we call the 
corresponding page, row and column bands together as the `rectangle'.
For each rectangle, we can now sample the `indicator pair'
$(k_1, k_2)(m_0, m_1, m_2) \in [2^{r_1}] \times [2^{r_2}]$ according to 
the random variable $(K_1, K_2)$ conditioned on the contents of the 
rectangle as described in the distribution $P_2$ of
Fact~\ref{fact:mutualcovering}. 
The full description of the random
codebook $\cC$ consists of the pages, symbols and indicator pairs. 
Given the codebook $\cC$, consider its {\em augmentation} $\cC'$
obtained by additionally choosing independent and uniform samples 
$l_0$, $l_1$, $l_2$  of computational basis vectors of $\cL$ to 
populate all the pages and
the rows and columns of $\cC$. We shall henceforth
work with the augmented codebook $\cC'$, which 
is revealed to Alice, Bob and Charlie.

\medskip

\noindent
{\bf Encoding:}

\noindent
To send message triple $(m_0, m_1, m_2)$, Alice picks up the entry
$\cC(m_0, m_1, k_1, m_2, k_2)$ where $(k_1, k_2)$ is the indicator pair
for the rectangle $(m_0, m_1, m_2)$. 
She then inputs the quantum state
$\sigma^X_{u_0(m_0), u_1(m_0, m_1, k_1), u_2(m_0, m_2, k_2)}$ into 
the channel $\chan$.

\medskip

\noindent
{\bf Decoding:}

\noindent
Consider the marginal cq-state $\rho^{U_0 U_1 Y_1}$. 
Express it as 
\[
\rho^{U_0 U_1 Y_1} =
\sum_{u_0, u_1} p(u_0, u_1) \ketbra{u_0, u_1}^{U_0 U_1}
\otimes \rho_{u_0, u_1}^{Y_1}.
\]
Define the cq-states
$\rho^{U_0 U_1 Y_1}_{(\{U_0\}, \{U_1\}, \{\})}$,
$\rho^{U_0 U_1 Y_1}_{(\{\}, \{U_0, U_1\}, \{\})}$ as in Claim~4 of
Fact~\ref{fact:cqtypical}.
Fix $0 < \delta < 1$. 
Fact~\ref{fact:cqtypical} tells us that there is an augmentation of
the classical systems
$U_0$, $U_1$ to $U'_0 := U_0 \otimes \cL$, $U'_1 := U_1 \otimes \cL$
an extension $Y'_1$ of the quantum system $Y_1$ i.e.
$Y_1 \otimes \C^2 \leq Y'_1$, a cq-state
$(\rho')^{U'_0 U'_1 Y'_1}$ and a cq-POVM element 
$(\Pi')^{U'_0 U'_1 Y'_1}$
such that
\begin{enumerate}

\item
$
(\rho')^{U'_0 U'_1 Y'_1}  =  
|\cL|^{-2}
\sum_{u_0, u_1, l_0, l_1} p(u_0, u_1) 
\ketbra{u_0, u_1}^{U_0 U_1} \otimes
\ketbra{l_0, l_1}^{\cL^{\otimes 2}} \otimes
(\rho')_{u_0, l_0, u_1, l_1}^{Y'_1},
$
for some quantum states $(\rho')_{u_0, l_0, u_1, l_1}^{Y'_1}$;

\item
For some POVM elements $(\Pi')_{u_0, l_0, u_1, l_1, \delta}^{Y'_1}$,
\begin{eqnarray*}
\lefteqn{(\Pi')^{U'_0 U'_1 Y'_1}} \\
& = &
\sum_{u_0, u_1, l_0, l_1} 
\ketbra{u_0, u_1}^{U_0 U_1} \otimes
\ketbra{l_0, l_1}^{\cL^{\otimes 2}} \\
&   &
~~~~~~~~~~~~~~~~
{} \otimes
(\Pi')_{u_0, l_0, u_1, l_1, \delta}^{Y'_1};
\end{eqnarray*}

\item
$
\ellone{
(\rho')^{U'_0 U'_1 Y'_1} -
\rho^{U_0 U_1 Y_1} \otimes 
\ketbra{0}^{\C^2} \otimes
\frac{\one^{\cL^{\otimes 2}}}{|\cL|^2} 
} \leq
8 \delta;
$

\item
$
\Tr [(\Pi')^{U'_0 U'_1 Y'_1} (\rho')^{U'_0 U'_1 Y'_1}] \geq
1 -
2^8 \cdot 3 \cdot \delta^{-2} \epsilon - 8 \delta;
$

\item
Define 
$(\rho')^{U'_0 U'_1 Y'_1}_{(\{U'_0\}, \{U'_1\}, \{\})}$,
$(\rho')^{U'_0 U'_1 Y'_1}_{(\{\}, \{U'_0, U'_1\}, \{\})}$ 
analogous to the states
$\rho^{U_0 U_1 Y_1}_{(\{U_0\}, \{U_1\}, \{\})}$,
$\rho^{U_0 U_1 Y_1}_{(\{\}, \{U_0, U_1\}, \{\})}$.
Then,
\[
\Tr [(\Pi')^{U'_0 U'_1 Y'_1} 
(\rho')^{U'_0 U'_1 Y'_1}_{(\{U'_0\}, \{U'_1\}, \{\})}
] \leq
2^{-I^{\epsilon}_H(U_1 : Y_1 | U_0)_\rho},
\]
\[
\Tr [(\Pi')^{U'_0 U'_1 Y'_1} 
(\rho')^{U'_0 U'_1 Y'_1}_{(\{\}, \{U'_0, U'_1\}, \{\})}
] \leq
2^{-I^{\epsilon}_H(U_0 U_1 : Y_1)_\rho}.
\]
\end{enumerate}

For $(m_0, m_1, k_1) \in [2^{R_0}] \times [2^{R_1 + r_1}]$, define 
the POVM element $\Lambda_{m_0, m_1, k_1}^{Y_1}$ 
as follows: Attach an ancilla of $\ketbra{0}^{\C^2}$
to register $Y_1$ and then apply POVM element 
$\Lambda_{(u_0, l_0)(m_0), (u_1, l_1)(m_0, m_1, k_1)}^{Y'_1}$.
Here $\Lambda_{(u_0, l_0)(m_0), (u_1, l_1)(m_0, m_1, k_1)}^{Y'_1}$
is a POVM element from the PGM
constructed, for the augmented codebook $\cC'$, from the set of 
positive operators
\[
\{
(\Pi')_{(u_0, l_0)(m'_0), (u_1, l_1)(m'_0, m'_1, k'_1), \delta}^{Y'_1}:
m'_0 \in 2^{R_0}, (m'_1, k'_1) \in [2^{R_1 + r_1}]
\},
\]
which in turn is provided by Claim~2 above.
Observe that $\Lambda_{m_0, m_1, k_1}^{Y_1}$ depends only on 
$(u_0, l_0)(m'_0)$, $(u_1, l_1)(m'_0, m'_1, k'_1)$,
$m'_0 \in 2^{R_0}$, $(m'_1, k'_1) \in [2^{R_1 + r_1}]$ of $\cC'$.
Similarly for $(m_0, m_2, k_2) \in [2^{R_0}] \times [2^{R_2 + r_2}]$,
we can define the POVM element 
$\Lambda_{m_0, m_2, k_2}^{Y_2}$.
Bob applies his POVM to the contents of $Y_1$ 
and outputs the result $(\hat{m}_0, \hat{m}_1, \hat{k}_1)$ as his guess 
for $(m_0, m_1, k_1)$.
Similarly, Charlie outputs 
$(\hat{\hat{m}}_0, \hat{\hat{m}}_2, \hat{\hat{k}}_2)$ as 
his guess for $(m_0, m_2, k_2)$. Bob and Charlie thus attempt to do 
the tougher job of decoding their
respective actual symbols inputted into the channel instead of 
just `decoding up to the band'.
\begin{figure*}
\begin{center}
\includegraphics[width=\textwidth]{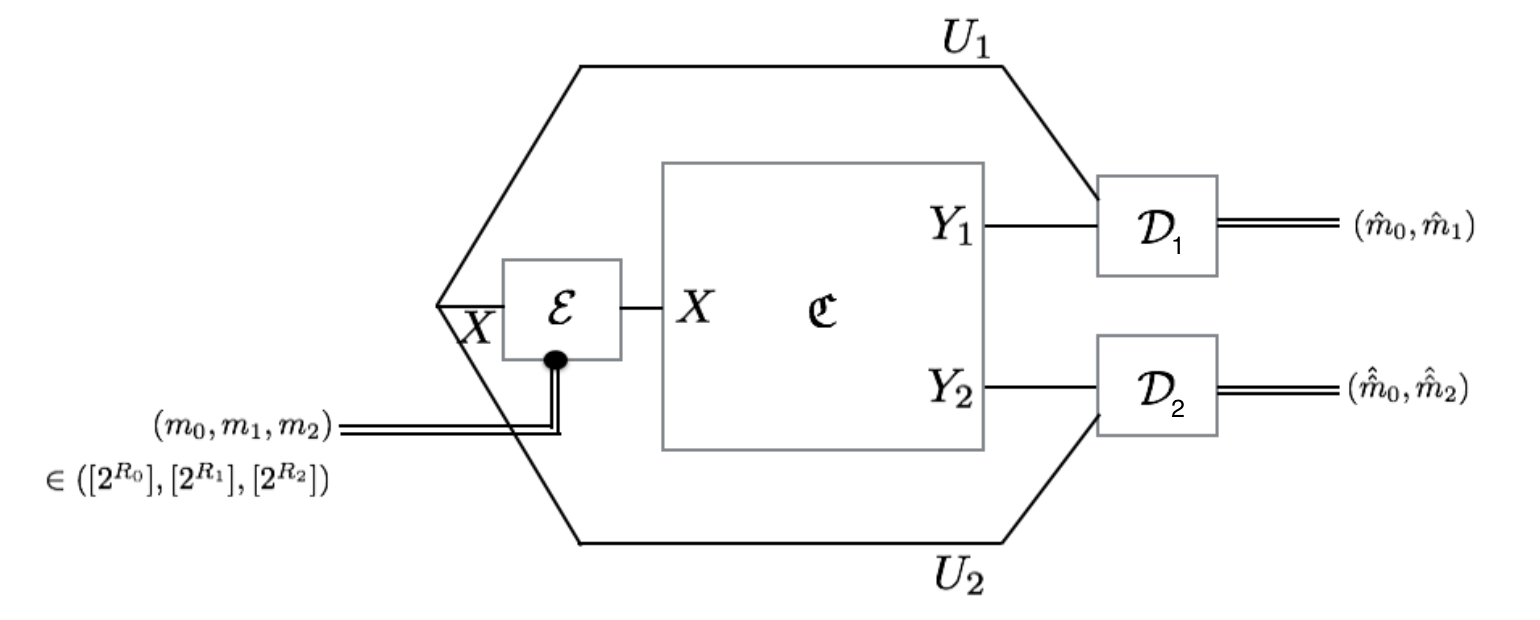}
\end{center}
\caption{Quantum broadcast channel with entanglement assistance.}
\label{fig:broadcastassisted}
\end{figure*}

\medskip

\noindent
{\bf Error probablity:}

\noindent
Suppose Alice transmits $(m_0, m_1, m_2)$. We consider the
expected decoding error of Bob over the choice of a random augmented
codebook $\cC'$.
We first observe that
by Fact~\ref{fact:mutualcovering}, at the cost of an additive decoding
error of $2 \epsilon$, we can pretend that we have
the distribution 
$(P_1)^{U_0 (U_1)^{2^{r_1}} (U_2)^{2^{r_2}} K_1 K_2}$
instead of the actual distribution 
$(P_2)^{U_0 (U_1)^{2^{r_1}} (U_2)^{2^{r_2}} K_1 K_2}$
inside rectangle $(m_1, m_2)$ of page $m_0$ of $\cC$. 
In other words, we can pretend that we first choose a 
uniformly random $(k_1, k_2) \in [2^{r_1}] \times [2^{r_2}]$,
put the cq-state
$\rho^{U_0 U_1 U_2 Y_1}$ between cell $(k_1, k_2)$ of 
rectangle $(m_1, m_2)$ of
page $m_0$ and Bob's output register $Y_1$, and independent copies
of $U_1 | U_0$, $U_2 | U_0$ in the other rows and columns 
of page $m_0$. In other pages, we
continue to have independent samples from the random variables
$U_0$, $U_1 | U_0$, $U_2 | U_0$. In all rectangles other than
rectangle $(m_1, m_2)$ of page $m_0$, we choose the indicator pairs
as described above during
the construction of the codebook $\cC$. We call the modified 
construction of the codebook as $\cC^{m_0, m_1, m_2, k_1, k_2}$ and
its augmentation as $(\cC')^{m_0, m_1, m_2, k_1, k_2}$.
This explains the inequality in Step~(a) below.
Next by Fact~\ref{fact:cqtypical},
at further cost of an additive decoding error of $4 \delta$ 
we shall pretend that we have the cq-state 
$(\rho')^{U'_0 U'_1 Y'_1}$ instead of 
$\rho^{U_0 U_1 Y_1}$
between cell $(k_1, k_2)$ of rectangle $(m_1, m_2)$ of page $m_0$ and
register $Y'_1$. Combining this with Fact~\ref{fact:HN} 
explains the inequality in Step~(b) below.
The inequality in Step~(c) below follows by an
application of Fact~\ref{fact:cqtypical}. 
We thus finally manage to bound Bob's expected decoding error.

For a state
$\sigma_{u_0(m_0), u_1(m_0, m_1, k_1), u_2(m_0, m_2, k_2)}^X$ inputted 
to the channel $\chan$, let
$\rho_{u_0(m_0), u_1(m_0, m_1, k_1), u_2(m_0, m_2, k_2)}^{Y_1}$
denote its output state at Bob's end.
We can bound Bob's expected decoding error as follows:
\begin{eqnarray*}
\lefteqn{
\E_{\cC'}[\Pr[\mbox{Bob's error}]]
} \\
&   =  &
\E_{\cC'}[
\Tr [
(\one^{Y_1} - \Lambda^{Y_1}_{m_0, m_1, k_1})
\rho_{u_0(m_0), u_1(m_0, m_1, k_1), u_2(m_0, m_2, k_2)}^{Y_1}
]
] \\
& \stackrel{a}{\leq} &
2 \epsilon + {} \\
&      &
2^{-r_1 - r_2} \sum_{k_1, k_2}
\E_{(\cC')^{m_0, m_1, m_2, k_1, k_2}}[ \\
&      &
~~~~~~~~~~~~~~~~~~~~~~~~~~~
\Tr [
(\one^{Y_1} - \Lambda^{Y_1}_{m_0, m_1, k_1}) \\
&      &
~~~~~~~~~~~~~~~~~~~~~~~~~~~~~~~~~~~~~
\rho_{u_0(m_0), u_1(m_0, m_1, k_1), u_2(m_0, m_2, k_2)}^{Y_1}
]
] \\
&   =  &
2 \epsilon + {}  \\
&      &
2^{-r_1 - r_2} \sum_{k_1, k_2}
\E_{(\cC')^{m_0, m_1, m_2, k_1, k_2}}[ \\
&      &
~~~~~~~~~~~~~~~~~~~~~~~~~~~~
\Tr [
(\one^{Y_1} - \Lambda^{Y_1}_{m_0, m_1, k_1}) \\
&      &
~~~~~~~~~~~~~~~~~~~~~~~~~~~~~~~~~~~~~~~~~
\rho_{u_0(m_0), u_1(m_0, m_1, k_1)}^{Y_1}
]
] \\
&   =  &
2 \epsilon + {} \\
&      &
2^{-r_1 - r_2} \sum_{k_1, k_2}
\E_{(\cC')^{m_0, m_1, m_2, k_1, k_2}}[ \\
&      &
~~~~~~~~~~~~~~~~~~~~~~~~~~~~
\Tr [
(\one^{Y'_1} - 
 \Lambda^{Y'_1}_{(u_0, l_0)(m_0), (u_1, l_1)(m_0, m_1, k_1)}) \\
&      &
~~~~~~~~~~~~~~~~~~~~~~~~~~~~~~~~~~~~~~~~~
(
\rho_{u_0(m_0), u_1(m_0, m_1, k_1)}^{Y_1} \otimes
\ketbra{0}^{\C^2} 
)
]
] \\
& \leq &
2 \epsilon + {} \\
&      &
2^{-r_1 - r_2} \sum_{k_1, k_2}
\E_{(\cC')^{m_0, m_1, m_2, k_1, k_2}}[ \\
&      &
~~~~~~~~~~~~~~~~~~~~~~~~~~~~
\Tr [
(\one^{Y'_1} - 
 \Lambda^{Y'_1}_{(u_0, l_0)(m_0), (u_1, l_1)(m_0, m_1, k_1)}) \\
&      &
~~~~~~~~~~~~~~~~~~~~~~~~~~~~~~~~~~~~~~~~~
(\rho')_{(u_0, l_0)(m_0), (u_1, l_1)(m_0, m_1, k_1)}^{Y'_1} 
] \\
&    &
{} + 
2^{-r_1 - r_2} \sum_{k_1, k_2}
\E_{(\cC')^{m_0, m_1, m_2, k_1, k_2}}[ \\
&      &
~~~~~~~~~~~~~~~~~~~~~~~~~~~~~~~~
\frac{1}{2}
\left\|
(\rho')_{(u_0, l_0)(m_0), (u_1, l_1)(m_0, m_1, k_1)}^{Y'_1}
\right. \\
&      &
~~~~~~~~~~~~~~~~~~~~~~~~~~~~~~~~~~~~~~~~~
\left.
{} -
\rho_{u_0(m_0), u_1(m_0, m_1, k_1)}^{Y_1} \otimes \ketbra{0}^{\C^2} 
\right\|_1 \\
&   =  &
2 \epsilon + {} \\
&      &
2^{-r_1 - r_2} \sum_{k_1, k_2}
\E_{(\cC')^{m_0, m_1, m_2, k_1, k_2}}[ \\
&      &
~~~~~~~~~~~~~~~~~~~~~~~~~~~~
\Tr [
(\one^{Y'_1} - 
 \Lambda^{Y'_1}_{(u_0, l_0)(m_0), (u_1, l_1)(m_0, m_1, k_1)}) \\
&      &
~~~~~~~~~~~~~~~~~~~~~~~~~~~~~~~~~~~~~~~~~
(\rho')_{(u_0, l_0)(m_0), (u_1, l_1)(m_0, m_1, k_1)}^{Y'_1} 
] \\
&   &
{} + \frac{1}{2} \ellone{
(\rho')^{U'_0 U'_1 Y'_1} -
\rho^{U_0 U_1 Y_1} \otimes \ketbra{0}^{\C^2} \otimes 
\frac{\one^{\cL^{\otimes 2}}}{|\cL|^2}
} \\
& \stackrel{b}{\leq} &
2 \epsilon + 4 \delta \\
&  &
{} + 2^{-r_1 - r_2} \cdot 2 \sum_{k_1, k_2}
\E_{(\cC')^{m_0, m_1, m_2, k_1, k_2}}[ \\
&      &
~~~~~~~~~~~~~~~~~~~~~~~~~~~~
\Tr [
(\one^{Y'_1} - 
 (\Pi')^{Y'_1}_{(u_0, l_0)(m_0), (u_1, l_1)(m_0, m_1, k_1)}) \\
&      &
~~~~~~~~~~~~~~~~~~~~~~~~~~~~~~~~~~~~~~~~~
(\rho')_{(u_0, l_0)(m_0), (u_1, l_1)(m_0, m_1, k_1)}^{Y'_1} 
] \\
&  &
{} +
2^{-r_1 - r_2} \cdot 4 \sum_{k_1, k_2}
\sum_{(m'_1, k'_1): (m'_1, k'_1) \neq (m_1, k_1)} \\
&   &
~~~~
\E_{(\cC')^{m_0, m_1, m_2, k_1, k_2}}[
\Tr [
(\Pi')^{Y'_1}_{(u_0, l_0)(m_0), (u'_1, l'_1)(m_0, m'_1, k'_1)} \\
&      &
~~~~~~~~~~~~~~~~~~~~~~~~~~~~~~~~~~~~~
(\rho')_{(u_0, l_0)(m_0), (u_1, l_1)(m_0, m_1, k_1)}^{Y'_1} 
]
] \\
&  &
{} +
2^{-r_1 - r_2} \cdot 4 \sum_{k_1, k_2}
\sum_{m'_0, m'_1, k'_1: m'_0 \neq m_0}  \\
&  &
~~~~
\E_{(\cC')^{m_0, m_1, m_2, k_1, k_2}}[
\Tr [
(\Pi')^{Y'_1}_{(u'_0, l'_0)(m'_0), (u'_1, l'_1)(m'_0, m'_1, k'_1)} \\
&      &
~~~~~~~~~~~~~~~~~~~~~~~~~~~~~~~~~~~~~
(\rho')_{(u_0, l_0)(m_0), (u_1, l_1)(m_0, m_1, k_1)}^{Y'_1} 
]
] \\
&   =  &
2 \epsilon + 4 \delta \\
&  &
{} + 
2 |\cL|^{-2} \sum_{u_0, u_1, l_0, l_1}
p(u_0, u_1)
\Tr [
(\one^{Y'_1} - (\Pi')^{Y'_1}_{u_0, u_1, l_0, l_1}) \\
&      &
~~~~~~~~~~~~~~~~~~~~~~~~~~~~~~~~~~~~~~~~~~~~~~~~~~~~~
(\rho')_{u_0, u_1, l_0, l_1}^{Y'_1} 
] \\
&  &
{} +
4 \cdot (2^{R_1 + r_1} - 1) 
|\cL|^{-3} \\
&  &
~~~~~~~~
\sum_{u_0, l_0, u_1, l_1, u'_1, l'_1} 
p(u_0) p(u_1 | u_0) p(u'_1 | u_0) \\
&  &
~~~~~~~~~~~~~~~~~~~~~~~~~~~~~~~
\Tr [
(\Pi')^{Y'_1}_{u_0, u'_1, l_0, l'_1}
(\rho')_{u_0, u_1, l_0, l_1}^{Y'_1} 
] \\
&  &
{} +
4 \cdot (2^{R_0} - 1) 2^{R_1 + r_1} 
|\cL|^{-4} \\
&  &
~~~~~~~~
\sum_{u_0, l_0, u'_0, l'_0, u_1, l_1, u'_1, l'_1}
p(u_0, u_1) p(u'_0, u'_1) \\
&  &
~~~~~~~~~~~~~~~~~~~~~~~~~~~~~~~~~~~~
\Tr [
(\Pi')^{Y'_1}_{u'_0, u'_1, l'_0, l'_1}
(\rho')_{u_0, u_1, l_0, l_1}^{Y'_1} 
] \\
&   =  &
2 \epsilon + 4 \delta + 
2 \Tr [
(\one^{U'_0 U'_1 Y'_1} - (\Pi')^{U'_0 U'_1 Y'_1})
(\rho')^{U'_0 U'_1 Y'_1} 
] \\
&  &
{} +
4 \cdot (2^{R_1 + r_1} - 1) 
\Tr [
(\Pi')^{U'_0 U'_1 Y'_1}
(\rho')^{U'_0 U'_1 Y'_1}_{(\{U'_0\}, \{U'_1\}, \{\})}
] \\
&  &
{} +
4 \cdot (2^{R_0} - 1) 2^{R_1 + r_1} 
\Tr [
(\Pi')^{U'_0 U'_1 Y'_1}
(\rho')^{U'_0 U'_1 Y'_1}_{(\{\}, \{U'_0, U'_1\}, \{\})}
] \\
& \stackrel{c}{\leq} &
2 \epsilon + 4 \delta + 
2^9 \cdot 3 \cdot \delta^{-2} \epsilon + 16 \delta \\
&  &
{} +
2^{R_1 + r_1 + 2 - I^\epsilon_H(U_1 : Y_1 | U_0)} +
2^{R_0 + R_1 + r_1 + 2 - I^\epsilon_H(U_0 U_1 : Y_1)}.
\end{eqnarray*}
Setting $\delta := \epsilon^{1/3}$, we get that 
$
\E_{\cC'}[\Pr[\mbox{Bob's error}]] \leq
2^{11} \epsilon^{1/3}.
$

Similarly, 
$
\E_{\cC'}[\Pr[\mbox{Charlie's error}]] \leq
2^{11} \epsilon^{1/3}.
$
Thus, there is an augmented codebook $\cC'$ such that sum of Bob's and
Charlie's average decoding errors is at most $2^{12} \epsilon^{1/3}$.
The average probability that at least one of Bob or Charlie err for
$\cC'$ is thus seen 
to be at most $2^7 \epsilon^{1/6}$ using Fact~\ref{fact:gentle}.
This finishes the proof of one-shot Marton's inner bound with 
common message.
\end{proof}

A similar proof as above combined with {\em position based coding}
technique of Anshu, Jain and Warsi~\cite{anshu:broadcast} can
be used to obtain a one-shot Marton's inner bound with common
message for sending classical information through an entanglement
assisted broadcast channel (see Figure~\ref{fig:broadcastassisted}). 
Earlier, Anshu, Jain and Warsi~\cite{anshu:broadcast} 
had shown the
achievability of a one-shot Marton's bound without common message.
\begin{theorem}[Ent. assist. one-shot Marton, com. msg.]
\label{thm:martonassisted}
Let $\chan: X \rightarrow Y_1 Y_2$ be a quantum broadcast channel. 
Let $\cU_0, \cU_1$, $\cU_2$ be
three new Hilbert spaces and $\psi^{U_0 U_1 U_2 X}$ be a quantum state
which is classical on $U_0$.
Consider the classical quantum state
\begin{eqnarray*}
\lefteqn{\rho^{U_0 U_1 U_2 Y_1 Y_2}} \\ 
& := &
\sum_{u_0} p(u_0)
\ketbra{u_0}^{U_0} \otimes {} \\
&   &
~~~~~~
((\chan^{X \rightarrow Y_1 Y_2} \otimes \I^{U_1 U_2})(
\psi^{U_1 U_2 X}_{u_0}))^{U_1 U_2 Y_1 Y_2}.
\end{eqnarray*}
Let $R_0$, $R_1$, $R_2$,  $\epsilon$ be such that
\begin{eqnarray*}
R_0 + R_1 
& \leq &
I_H^{\epsilon}(U_0 U_1 : Y_1) - 2 -  \log \frac{1}{\epsilon} \\
R_0 + R_2 
& \leq &
I_H^{\epsilon}(U_0 U_2 : Y_2) - 2 - \log \frac{1}{\epsilon} \\
R_0 + R_1 + R_2 
& \leq &
I_H^{\epsilon}(U_0 U_2 : Y_2) +
I_H^{\epsilon}(U_1 : Y_1 | U_0) \\
&      &
{} - 
I_\infty^{\epsilon, \epsilon^2}(U_1 : U_2 | U_0) 
- 4 - 4 \log \frac{1}{\epsilon} \\
R_0 + R_1 + R_2 
& \leq &
I_H^{\epsilon}(U_0 U_1 : Y_1) +
I_H^{\epsilon}(U_2 : Y_2 | U_0) \\
&      &
{} - 
I_\infty^{\epsilon, \epsilon^2}(U_1 : U_2 | U_0) 
- 4 - 4 \log \frac{1}{\epsilon} \\
2 R_0 + R_1 + R_2 
& \leq &
I_H^{\epsilon}(U_0 U_1 : Y_1) +
I_H^{\epsilon}(U_0 U_2 : Y_2) \\
&      &
{} - 
I_\infty^{\epsilon, \epsilon^2}(U_1 : U_2 | U_0) 
- 4 - 4 \log \frac{1}{\epsilon},
\end{eqnarray*}
where the mutual
information quantitites above are computed with respect to the
cq-state $\rho^{U_0 U_1 U_2 Y_1 Y_2}$. 
Then there exists an 
$(R_0, R_1, R_2, 2^7 \epsilon^{1/10})$-quantum broadcast 
channel code for sending
classical information through $\chan$ with entanglement assistance. 
\end{theorem}
\begin{proof}
Let $R_0$, $R_1$, $R_2$, $r_1$, $r_2$, $\epsilon$ be such that
\begin{eqnarray*}
R_0 + R_1 + r_1
& \leq &
I_H^{\epsilon}(U_0 U_1 : Y_1) - 2 - \log \frac{1}{\epsilon} \\
R_0 + R_2 + r_2
& \leq &
I_H^{\epsilon}(U_0 U_2 : Y_2) - 2 - \log \frac{1}{\epsilon} \\
R_1 + r_1
& \leq &
I_H^{\epsilon}(U_1 : Y_1 | U_0) - 2 - \log \frac{1}{\epsilon} \\
R_2 + r_2
& \leq &
I_H^{\epsilon}(U_2 : Y_2 | U_0) - 2 - \log \frac{1}{\epsilon} \\
r_1 + r_2
& = &
I_\infty^{\epsilon,\epsilon^2}(U_1 : U_2 | U_0) + 
2 \log \frac{1}{\epsilon},
\end{eqnarray*}
where the mutual
information quantitites above are computed with respect to the
cq-state $\rho^{U_0 U_1 U_2 Y_1 Y_2}$. 
Suppose we show that there is a
$(R_0, R_1, R_2, 2^7 \epsilon^{1/10})$-quantum broadcast channel code 
for sending classical information through $\chan$. 
Standard Fourier-Motzkin elimination
can be used to get rid of $r_1$ and $r_2$ and obtain the inner bound
in the statement of the theorem.

\medskip

\noindent
{\bf Codebook:}

\noindent
The codebook $\cC$ is now classical-quantum. It has $2^{R_0}$ {\em pages} 
and is generated randomly
as follows. First sample 
\[
u_0(1), \ldots, u_0(2^{R_0})
\]
independently according
to $\psi^{U_0}$. We will associate $u_0(m_0)$ with the $m_0$th page of 
$\cC$. 
Now, to generate the contents of the $m_0$th page, take
independent copies 
\[
\psi^{U_1(m_0, 1) U'_1(m_0, 1)}|m_0, \ldots, 
\psi^{U_1(m_0, 2^{R_1+r_1}) U'_1(m_0, 2^{R_1+r_1})}|m_0, 
\]
\[
\psi^{U_2(m_0, 1) U'_2(m_0, 1)}|m_0, \ldots, 
\psi^{U_2(m_0, 2^{R_2+r_2}) U'_2(m_0, 2^{R_2+r_2})}|m_0, 
\]
of the states $\psi^{U_1 U'_1}|m_0$, $\psi^{U_2 U'_2}|m_0$,
where $\psi^{U_1 U'_1}|m_0$ is a purification
of $\psi^{U_1}|m_0$, the marginal state on $U_1$ obtained by conditioning
the register $U_0$ in $\psi^{U_0 U_1}$ to take the value $m_0$,
and $\psi^{U_2 U'_2}|m_0$ is defined similarly. The registers 
\[
U'_1(m_0, 1), \ldots, U'_1(m_0, 2^{R_1+r_1})
\]
are with Bob,
\[
U'_2(m_0, 1), \ldots, U'_2(m_0, 2^{R_2+r_2})
\]
with Charlie, and
\[
U_1(m_0, 1), \ldots, U_1(m_0, 2^{R_1+r_1}),
U_2(m_0, 1), \ldots, U_2(m_0, 2^{R_2+r_2})
\]
with Alice. The states
\[
\psi^{U_1(m_0, 1) U'_1(m_0, 1)}|m_0, \ldots, 
\psi^{U_1(m_0, 2^{R_1+r_1}) U'_1(m_0, 2^{R_1+r_1})}|m_0 
\]
form the prior entanglement between Alice and Bob, and the states
\[
\psi^{U_2(m_0, 1) U'_2(m_0, 1)}|m_0, \ldots, 
\psi^{U_2(m_0, 2^{R_2+r_2}) U'_2(m_0, 2^{R_2+r_2})}|m_0
\]
form the prior entanglement between Alice and Charlie.
For $(m_0, m_1) \in [2^{R_0}] \times [2^{R_1}]$, define the 
`row band' $\cC(m_0, m_1)$ to be the states
\[
\begin{array}{l}
\psi^{U_1(m_0, (m_1 - 1) 2^{r_1} + 1) 
      U'_1(m_0, (m_1 - 1) 2^{r_1} + 1)}|m_0, \\
~~~~~~~
{} \ldots, 
\psi^{U_1(m_0, m_1 2^{r_1}) U'_1(m_0, m_1 2^{r_1})}|m_0
\end{array}
\]
page $m_0$.
Similarly, one can define the `column band' $\cC(m_0, m_2)$ for each
$(m_0, m_2) \in [2^{R_0}] \times [2^{R_2}]$. 
For a triple $(m_0, m_1, m_2)$, we call the  
corresponding page, row and column bands together as the `rectangle'.

For each rectangle, we can now sample the `indicator pair'
$(k_1, k_2)(m_0, m_1, m_2) \in [2^{r_1}] \times [2^{r_2}]$ according to 
the random variable $(K_1, K_2)$ arising from an application of
the bipartite convex
split lemma of \cite{anshu:broadcast} used with underlying state
$\psi^{U_1 U_2 X}|m_0$.
The sampling process creates as side effect a quantum register that 
we call a
`candidate channel input register'
$X(m_0, m_1, k_1, m_2, k_2)$. The resulting state on the
registers 
\[
\begin{array}{l}
X(m_0, m_1, k_1, m_2, k_2) \\
~~~~
U_1(m_0, (m_1 - 1) 2^{r_1} + 1) 
U'_1(m_0, (m_1 - 1) 2^{r_1} + 1)  \\
~~~~~~~
U_2(m_0, (m_2 - 1) 2^{r_2} + 1) 
U'_2(m_0, (m_2 - 1) 2^{r_2} + 1) 
\end{array}
\]
is $(8\epsilon)$-close to
a purification $\psi^{U_1 U'_1 U_2 U'_2 X}|m_0$ of 
$\psi^{U_1 U_2 X}|m_0$.
For more details, see \cite{anshu:broadcast}.
The full description of the random
codebook $\cC$ consists of the pages, prior entanglement, indicator 
pairs and candidate channel input registers.
Given the codebook $\cC$, consider its {\em augmentation} $\cC'$
obtained by additionally choosing independent and uniform samples 
$l_0$, $l_1$, $l_2$  of computational basis vectors of $\cL$ to 
populate all the pages and
the rows and columns of $\cC$. We shall henceforth
work with the augmented codebook $\cC'$, which 
is revealed to Alice, Bob and Charlie.

\medskip

\noindent
{\bf Encoding:}

\noindent
To send message triple $(m_0, m_1, m_2)$, Alice picks up 
the indicator pair $(k_1, k_2)$ for the rectangle $(m_0, m_1, m_2)$. 
She then inputs the register $X(m_0, m_1, k_1, m_2, k_2)$
into the channel $\chan$.

\medskip

\noindent
{\bf Decoding:}

\noindent
Consider the marginal cq-state $\rho^{U_0 U_1 Y_1}$. 
Express it as 
\[
\rho^{U_0 U_1 Y_1} =
\sum_{u_0} p(u_0) \ketbra{u_0}^{U_0}
\otimes \rho_{u_0}^{U_1 Y_1}.
\]
Define the cq-states
\begin{eqnarray*}
\rho^{U_0 U_1 Y_1}_{[\{U_0\}, (\{U_1\}, \{Y_1\})]}
& := &
\sum_{u_0} p(u_0) \ketbra{u_0}^{U_0}
\otimes \rho_{u_0}^{U_1} \otimes \rho_{u_0}^{Y_1}, \\
\rho^{U_0 U_1 Y_1}_{[\{\}, (\{U_0, U_1\}, \{Y_1\})]}
& := &
\left(\sum_{u_0} p(u_0) \ketbra{u_0}^{U_0} \otimes \rho_{u_0}^{U_1}\right)
\otimes \rho^{Y_1}.
\end{eqnarray*}
Fix $0 < \delta < 1$. 
The full version of the intersection case of the
classical quantum joint typicality lemma, 
viz. Lemma~1 from \cite{sen:oneshot}, 
tells us that there is an augmentation of
the classical system
$U_0$ to $U'_0 := U_0 \otimes \cL$, 
augmentations $U'_1 := U_1 \otimes \C^2 \otimes \cL$, 
$Y'_1 := Y_1 \otimes \C^2 \otimes \cL$ 
of the quantum systems $U_1$, $Y_1$, a cq-state
$(\rho')^{U'_0 U'_1 Y'_1}$ and a cq-POVM element 
$(\Pi')^{U'_0 U'_1 Y'_1}$
such that
\begin{enumerate}

\item
\begin{eqnarray*}
\lefteqn{(\rho')^{U'_0 U'_1 Y'_1}} \\  
& =  &
|\cL|^{-3}
\sum_{u_0, l_0, l_1, \hl_1} p(u_0) 
\ketbra{u_0}^{U_0} \otimes
\ketbra{l_0, l_1, \hl_1}^{\cL^{\otimes 3}} \\ 
&    &
~~~~~~~~~~~~~~~~~~~~~~~~~~~~~~~~~~~~~
{} \otimes
(\rho')_{u_0, l_0, l_1, \hl_1}^{U'_1 Y'_1},
\end{eqnarray*}
for some quantum states $(\rho')_{u_0, l_0, l_1, \hl_1}^{U'_1 Y'_1}$;

\item
For some POVM elements $(\Pi')_{u_0, l_0, l_1, \hl_1, \delta}^{Y'_1}$,
\begin{eqnarray*}
\lefteqn{(\Pi')^{U'_0 U'_1 Y'_1}} \\
& = &
\sum_{u_0, l_0, l_1, \hl_1} 
\ketbra{u_0}^{U_0 U_1} \otimes
\ketbra{l_0, l_1, \hl_1}^{\cL^{\otimes 3}} \\
&   &
~~~~~~~~~~~~~~~~
{} \otimes
(\Pi')_{u_0, l_0, l_1, \hl_1, \delta}^{Y'_1};
\end{eqnarray*}

\item
$
\ellone{
(\rho')^{U'_0 U'_1 Y'_1} -
\rho^{U_0 U_1 Y_1} \otimes 
\ketbra{0}^{\C^2} \otimes
\frac{\one^{\cL^{\otimes 3}}}{|\cL|^3} 
} \leq
8 \delta;
$

\item
$
\Tr [(\Pi')^{U'_0 U'_1 Y'_1} (\rho')^{U'_0 U'_1 Y'_1}] \geq
1 -
2^8 \cdot 3 \cdot \delta^{-4} \epsilon - 8 \delta;
$

\item
Define 
$(\rho')^{U'_0 U'_1 Y'_1}_{[\{U'_0\}, (\{U'_1\}, \{Y'_1\})]}$,
$(\rho')^{U'_0 U'_1 Y'_1}_{[\{\}, (\{U'_0, U'_1\}, \{Y'_1\})]}$ 
analogous to the states
$\rho^{U_0 U_1 Y_1}_{[\{U_0\}, (\{U_1\}, \{Y_1\})]}$,
$\rho^{U_0 U_1 Y_1}_{[\{\}, (\{U_0, U_1\}, \{Y_1\})]}$.
Then,
\[
\Tr [(\Pi')^{U'_0 U'_1 Y'_1} 
(\rho')^{U'_0 U'_1 Y'_1}_{[\{U'_0\}, (\{U'_1\}, \{Y'_1\})]}
] \leq
2^{-I^{\epsilon}_H(U_1 : Y_1 | U_0)_\rho},
\]
\[
\Tr [(\Pi')^{U'_0 U'_1 Y'_1} 
(\rho')^{U'_0 U'_1 Y'_1}_{[\{\}, (\{U'_0, U'_1\}, \{Y'_1\})]}
] \leq
2^{-I^{\epsilon}_H(U_0 U_1 : Y_1)_\rho}.
\]
\end{enumerate}

For $(m_0, m_1, k_1) \in [2^{R_0}] \times [2^{R_1 + r_1}]$, define 
the POVM element $\Lambda_{m_0, m_1, k_1}^{Y_1}$ according to
the position based decoding strategy of \cite{anshu:broadcast}
using the positive operators
\[
\{
(\Pi')_{(u_0, l_0)(m'_0), (l_1, \hl_1)(m'_0, m'_1, k'_1), 
        \delta}^{Y'_1}:
m'_0 \in 2^{R_0}, (m'_1, k'_1) \in [2^{R_1 + r_1}]
\},
\]
which in turn is provided by Claim~2 above.
Observe that $\Lambda_{m_0, m_1, k_1}^{Y_1}$ depends only on
$(u_0, l_0)(m'_0)$, $(l_1, \hl_1)(m'_0, m'_1, k'_1)$,
$m'_0 \in 2^{R_0}$, $(m'_1, k'_1) \in [2^{R_1 + r_1}]$ of $\cC'$.
Similarly for $(m_0, m_2, k_2) \in [2^{R_0}] \times [2^{R_2 + r_2}]$,
we can define the POVM element 
$\Lambda_{m_0, m_2, k_2}^{Y_2}$.
Bob applies his POVM to the contents of $Y_1$ 
and outputs the result $(\hat{m}_0, \hat{m}_1, \hat{k}_1)$ as his guess 
for $(m_0, m_1, k_1)$.
Similarly, Charlie outputs 
$(\hat{\hat{m}}_0, \hat{\hat{m}}_2, \hat{\hat{k}}_2)$ as 
his guess for $(m_0, m_2, k_2)$. Bob and Charlie thus attempt to do 
the tougher job of decoding their
respective actual symbols inputted into the channel instead of 
just `decoding up to the band'.
\begin{figure*}
\begin{center}
\includegraphics[width=\textwidth]{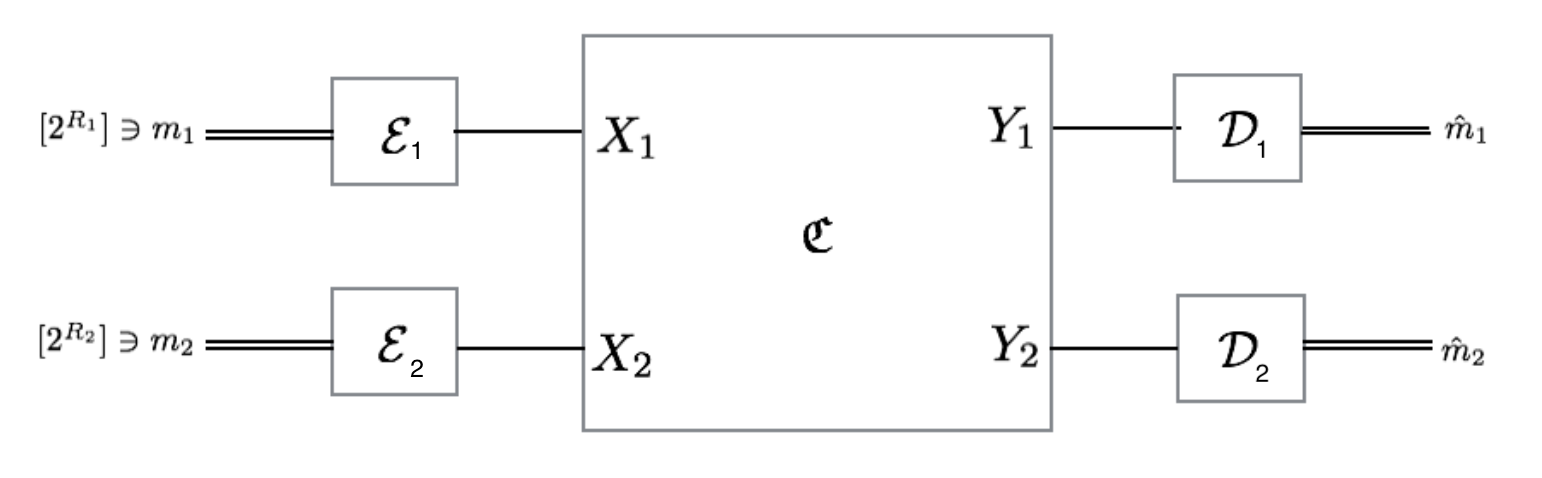}
\end{center}
\caption{Quantum interference channel without entanglement assistance.}
\label{fig:interference}
\end{figure*}

\medskip

\noindent
{\bf Error probablity:}

\noindent
The error probability calculation is very similar to that in the proof
of Theorem~\ref{thm:marton} above. This is because the error analysis
in the bipartite convex split lemma and position based decoding of
\cite{anshu:broadcast} is very similar to the error analysis
in our mutual covering lemma (Fact~\ref{fact:mutualcovering}) and
in pretty good measurement based decoding.

Setting $\delta := \epsilon^{1/5}$, we get that 
$
\E_{\cC'}[\Pr[\mbox{Bob's error}]] \leq
2^{11} \epsilon^{1/5}.
$

Similarly, 
$
\E_{\cC'}[\Pr[\mbox{Charlie's error}]] \leq
2^{11} \epsilon^{1/5}.
$
Thus, there is an augmented codebook $\cC'$ such that sum of Bob's and
Charlie's average decoding errors is at most $2^{12} \epsilon^{1/5}$.
The average probability that at least one of Bob or Charlie err for
$\cC'$ is thus seen 
to be at most $2^7 \epsilon^{1/10}$ using Fact~\ref{fact:gentle}.
This finishes the proof of the entanglement assisted one-shot 
Marton's inner bound with common message.
\end{proof}

\medskip

\noindent
{\bf Remark:}

\noindent
The above theorem is unsatisfactory as the state 
$\psi^{U_0 U_1 U_2 X}$ used therein is classical on $U_0$. 
This is because the inner bound expression in the theorem
contains one-shot mutual information terms that condition on $U_0$. 
No proper definition of these terms is known when $U_0$ is quantum.
This deficiency is further reflected in the statements of the 
classical-quantum joint typicality lemmas in
Facts~\ref{fact:cqtypical} and \ref{fact:gencqtypical} as well as in
their full versions in \cite{sen:oneshot}, all of which 
can only condition on classical registers.
On a different vein, observe that the register $U_0$ captures
the common message in the protocol. It is
unclear how to define a common message for a broadcast channel in 
the case of transmission
of quantum information, whereas the personal messages have straightforward
quantum analogues. This may also be another reason why we are unable
to make $U_0$ quantum in the statement of the theorem.
Making $\psi^{U_0 U_1 U_2 X}$ fully quantum thus remains an open problem. 

\section{Interference channel}
\label{sec:interference}
We now prove one-shot inner bounds for
sending classical information through a quantum interference channel
(q-IC). 
In this problem, there are two senders $A_1$, $A_2$ and their 
corresponding receivers $B_1$, $B_2$. Sender $A_1$ would like to
send a classical message $m_1 \in [2^{R_1}]$ to $B_1$. Similarly,
$A_2$ would like to send $m_2 \in [2^{R_2}]$ to $B_2$.
The parties have at their disposal
a quantum channel $\chan: X_1 X_2 \rightarrow Y_1 Y_2$ with input 
Hilbert spaces $\cX_1$, $\cX_2$ and 
output Hilbert spaces $\cY_1$, $\cY_2$. 
Sender $A_1$ encodes 
$m_1$ into a quantum state $\sigma_{m_0}^{X_1} \in \cX_1$ and inputs it 
to $\chan$. 
Similarly, $A_2$ encodes 
$m_2$ into a quantum state $\sigma_{m_2}^{X_2} \in \cX_2$ and inputs it 
to $\chan$. 
The channel outputs a quantum  state 
$
\rho_{m_1, m_2}^{Y_1 Y_2} := 
(\chan^{X_1 X_2 \rightarrow Y_1 Y_2}(
\sigma_{m_1}^{X_1} \otimes \sigma_{m_2}^{X_2}
)
)^{Y_1 Y_2}$ 
jointly
supported in the Hilbert space $\cY_1 \otimes \cY_2$. Receivers
$B_1$, $B_2$
apply their respective decoding superoperators independently on
$\rho_{m_1, m_2}^{Y_1 Y_2}$ in order to produce their respective guesses 
$\hat{m}_1$, $\hat{m}_2$ of the messages $m_1$, $m_2$. 
See Figure~\ref{fig:interference}.
Let $0 \leq \epsilon \leq 1$. Consider the uniform probability distribution
over the message sets.
We want that 
$
\Pr[(\hat{m}_1, \hat{m}_2) \neq 
    (m_1, m_2)] \leq \epsilon,
$
where the probability is over the choice of the messages and actions of the
encoder, channel and decoders. If there exist such encoding and
decoding schemes for a particular channel $\chan$, we say that there exists
an $(R_1, R_2, \epsilon)$-quantum interference channel code for 
sending classical information through $\chan$.

We now state and prove our one-shot Chong-Motani-Garg-El Gamal style inner 
bound for
sending classical information through an unassisted quantum interference
channel. Our inner bound reduces to the standard Chong-Motani-Garg-El Gamal
inner bound for the asymtotic iid setting, which is also known
to be equivalent to the famous Han-Kobayashi inner bound
\cite{CMGElGamal}.
However, in the one-shot setting it is unclear if the two inner bounds are
the same.
\begin{theorem}[One-shot Chong-Motani-Garg-El Gamal]
Let $\chan: X'_1 X'_2 \rightarrow Y_1 Y_2$ be a quantum interference 
channel. Let $\cQ$, $\cU_1$, $\cX_1$, $\cU_2$, $\cX_2$ be
four new sample spaces. Let the $4$-tuple $(Q, U_1, X_1, U_2, X_2)$ be a 
jointly distributed random variable with probability mass function
$p(q) p(u_1, x_1 | q) p(u_2, x_2 | q)$.
For every element $x_1 \in \cX_1$, $x_2 \in \cX_2$, let 
$\sigma_{x_1}^{X'_1}$, $\sigma_{x_2}^{X'_2}$
be quantum states in the input Hilbert spaces $\cX'_1$, $\cX'_2$ of 
$\chan$.
Consider the classical quantum state
\begin{eqnarray*}
\lefteqn{\rho^{Q U_1 X_1 U_2 X_2 Y_1 Y_2}} \\
& := &
\sum_{q, u_1, x_1, u_2, x_2}
p(q) p(u_1, x_1 | q) p(u_2, x_2 | q) \\
&    &
~~~~~~~~~~~~~~~~~~~~
\ketbra{q, u_1, x_1, u_2, x_2}^{Q U_1 X_1 U_2 X_2} \\
&   &
~~~~~~~~~~~~~~~~~~~~~~~~~~
{} \otimes
(\chan(\sigma_{x_1}^{X'_1} \otimes \sigma_{x_2}^{X'_2}))^{Y_1 Y_2}.
\end{eqnarray*}
Let $R_1$, $R_2$, $\epsilon$, 
be such that
\begin{eqnarray*}
R_1 
& \leq &
I_H^{\epsilon}(X_1 : Y_1 | U_2 Q) - 2 -  \log \frac{1}{\epsilon} \\
R_1 
& \leq &
I_H^{\epsilon}(X_1 : Y_1 | U_1 U_2 Q) +
I_H^{\epsilon}(X_2 U_1 : Y_2 | U_2 Q) \\
&      &
{} - 2 -  \log \frac{1}{\epsilon} \\
R_2 
& \leq &
I_H^{\epsilon}(X_2 : Y_2 | U_1 Q) - 2 -  \log \frac{1}{\epsilon} \\
R_2 
& \leq &
I_H^{\epsilon}(X_2 : Y_2 | U_1 U_2 Q) +
I_H^{\epsilon}(X_1 U_2 : Y_1 | U_1 Q) \\
&      &
{} - 2 -  \log \frac{1}{\epsilon} \\
R_1 + R_2 
& \leq &
I_H^{\epsilon}(X_1 U_2 : Y_1 | Q) +
I_H^{\epsilon}(X_2 : Y_2 | U_1 U_2 Q) \\
&      &
{} - 2 -  \log \frac{1}{\epsilon} \\
R_1 + R_2 
& \leq &
I_H^{\epsilon}(X_2 U_1 : Y_2 | Q) +
I_H^{\epsilon}(X_1 : Y_1 | U_1 U_2 Q) \\
&      &
{} - 2 -  \log \frac{1}{\epsilon} \\
R_1 + R_2 
& \leq &
I_H^{\epsilon}(X_1 U_2 : Y_1 | U_1 Q) +
I_H^{\epsilon}(X_2 U_1 : Y_2 | U_2 Q) \\
&      &
{} - 2 -  \log \frac{1}{\epsilon} \\
2 R_1 + R_2 
& \leq &
I_H^{\epsilon}(X_1 U_2 : Y_1 | Q) +
I_H^{\epsilon}(X_1 : Y_1 | U_1 U_2 Q) \\
&      &
{} +
I_H^{\epsilon}(X_2 U_1 : Y_2 | U_2 Q) - 2 -  \log \frac{1}{\epsilon} \\
R_1 + 2 R_2 
& \leq &
I_H^{\epsilon}(X_2 U_1 : Y_2 | Q) +
I_H^{\epsilon}(X_2 : Y_2 | U_1 U_2 Q) \\
&      &
{} +
I_H^{\epsilon}(X_1 U_2 : Y_1 | U_1 Q) - 2 -  \log \frac{1}{\epsilon}, 
\end{eqnarray*}
where the mutual
information quantitites above are computed with respect to the
cq-state $\rho^{Q U_1 X_1 U_2 X_2 Y_1 Y_2}$. 
Then there exists an 
$(R_1, R_2, 2^{2^{14}} \epsilon^{1/6})$-quantum interference channel code 
for sending
classical information through $\chan$. 
\end{theorem}
\begin{proof}
We follow the proof outline as given in El Gamal-Kim's 
book~\cite{book:elgamalkim}. 
We use `rate splitting' to divide $A_1$'s
message $m_1 \in [2^{R_1}]$ into a `public part' 
$m'_1 \in [2^{R'_1}]$ and a `personal part' $m''_1 \in [2^{R_1 - R'_1}]$.
Similarly, we divide $A_2$'s
message $m_2 \in [2^{R_2}]$ into a `public part' 
$m'_2 \in [2^{R'_2}]$ and a `personal part' $m''_2 \in [2^{R_2 - R'_2}]$.
The public messages must be recovered by both receivers whereas the
personal messages need only to be recovered by the intended receiver.
The messages are sent by a one-shot version of superposition coding 
whereby the `cloud centres' $u_1$, $u_2$ carry the public messages
$m'_1$, $m'_2$ and the `satellite symbols' $x_1$, $x_2$, which will 
be decoded after first recovering $u_1$, $u_2$, carry the personal
messages $m''_1$, $m''_2$.

We now show that a rate quadruple 
$(R'_1, R_1 - R'_1, R'_2, R_2 - R'_2)$ is achievable if it satisfies
the following inequalities.
\begin{eqnarray*}
R_1 - R'_1 
& \leq &
I_H^{\epsilon}(X_1 : Y_1 | U_1 U_2 Q) - 2 -  \log \frac{1}{\epsilon} \\
R_1 
& \leq &
I_H^{\epsilon}(X_1 : Y_1 | U_2 Q) - 2 -  \log \frac{1}{\epsilon} \\
R_1 - R'_1 + R'_2 
& \leq &
I_H^{\epsilon}(X_1 U_2 : Y_1 | U_1 Q) - 2 -  \log \frac{1}{\epsilon} \\
R_1 + R'_2 
& \leq &
I_H^{\epsilon}(X_1 U_2 : Y_1 | Q) - 2 -  \log \frac{1}{\epsilon} \\
R_2 - R'_2 
& \leq &
I_H^{\epsilon}(X_2 : Y_2 | U_1 U_2 Q) - 2 -  \log \frac{1}{\epsilon} \\
R_2 
& \leq &
I_H^{\epsilon}(X_2 : Y_2 | U_1 Q) - 2 -  \log \frac{1}{\epsilon} \\
R_2 - R'_2 + R'_1 
& \leq &
I_H^{\epsilon}(X_2 U_1 : Y_2 | U_2 Q) - 2 -  \log \frac{1}{\epsilon} \\
R_2 + R'_1 
& \leq &
I_H^{\epsilon}(X_2 U_1 : Y_2 | Q) - 2 -  \log \frac{1}{\epsilon} \\
\end{eqnarray*}
where the mutual
information quantitites above are computed with respect to the
cq-state $\rho^{Q U_1 X_1 U_2 X_2 Y_1 Y_2}$. Standard Fourier-Motzkin
elimination now gives us the rate region in the statement of the theorem.
Note that in the one-shot case it is not clear if the second upper
bounds on $R_1$ and $R_2$, in the rate region described in the 
theorem statement, can be eliminated, unlike the asymptotic iid case.
This is because their elimination in the asymptotic iid case relies 
on the chain rule for Shannon
mutual information, which is not known to hold  
for the hypothesis testing mutual information used in the one-shot
setting.

\medskip

\noindent
{\bf Codebook:}

\noindent
First generate a sample $q$ from the distribution $p(q)$. 
For each public message
$m'_1 \in [2^{R'_1}]$ independently generate a sample
$u_1(m'_1)$ from the distribution $p(u_1 | q)$.
Similarly, for each public message
$m'_2 \in [2^{R'_2}]$ independently generate a sample
$u_2(m'_2)$ from the distribution $p(u_2 | q)$.
Now for each public message $m'_1$, independently generate samples
$x_1(m'_1, m''_1)$ from the distribution $p(x_1 | u_1 q)$ for all
personal messages $m''_1 \in [2^{R_1 - R'_1}]$.
Similarly for each public message $m'_2$, independently generate samples
$x_2(m'_2, m''_2)$ from the probability distribution 
$p(x_2 | u_2 q)$ for all
personal messages $m''_2 \in [2^{R_2 - R'_2}]$. These samples together
consititute the random codebook $\cC$.
Given the codebook $\cC$, consider its {\em augmentation} $\cC'$
obtained by additionally choosing independent and uniform samples
$l_0$, $l'_1$, $l''_1$,  $l'_2$, $l''_2$  of computational basis 
vectors of $\cL$ to populate all the entries of $\cC'$. We shall henceforth
work with the augmented codebook $\cC'$, which 
is revealed to $A_1$, $A_2$, $B_1$, $B_2$.

\medskip

\noindent
{\bf Encoding:}

\noindent
To send message $m_1 = (m'_1, m''_1)$, $A_1$ picks up the symbol
$x_1(m'_1, m''_1)$  from the codebook $\cC$ and inputs the 
state $\sigma_{x_1(m'_1, m''_1)}^{X'_1}$ into the channel $\chan$.
Similarly, to send message $m_2 = (m'_2, m''_2)$, $A_2$ picks up the symbol
$x_2(m'_2, m''_2)$  from the codebook $\cC$ and inputs the 
state $\sigma_{x_2(m'_2, m''_2)}^{X'_2}$ into the channel $\chan$.

\medskip

\noindent
{\bf Decoding:}

\noindent
The receiver $B_1$ decodes the tuple $(m''_1, m'_1)$ using
simultaneous {\em non-unique} decoding. To do this, he has to apply
a `union of intersection' of POVM elements which in turn is provided
by Fact~\ref{fact:gencqtypical}. The `union' is over all choices of
$\hat{m}'_2 \in [2^{R'_2}]$. In the asymptotic iid setting it turns
out that non-unique decoding is not required in order to get the
Chong-Motani-Garg-El Gamal rate region. Sen~\cite{sen:interference} 
showed that we can further
require $B_1$ to recover $m'_2$ and still obtain the same rate 
region. However the argument in \cite{sen:interference} fails in 
the one-shot setting since it
relies on chain rule of Shannon mutual information which is not known to
hold for the hypothesis testing mutual information. Chain rules for
smooth one-shot mutual information quantities are typically inequalities
and frequently involve two or more types of quantities in 
the same expression. 
Hence using them often leads to unsatisfactory bounds for channel
coding problems.
Therefore
we use non-unique decoding in the one-shot setting as it possibly leads
to a larger inner bound.

Consider the marginal cq-state $\rho^{Q U_1 X_1 U_2 Y_1}$. 
Express it as 
\begin{eqnarray*}
\lefteqn{\rho^{Q U_1 X_1 U_2 Y_1}} \\ 
& = &
\sum_{q, u_1, x_1, u_2} p(q) p(u_1, x_1 | q) p(u_2 | q) \\
&   &
~~~~~~~~~~~~~
\ketbra{q, u_1, x_1, u_2}^{Q U_1 X_1 U_2}
\otimes \rho_{q, u_1, x_1, u_2}^{Y_1},
\end{eqnarray*}
where in fact $\rho_{q, u_1, x_1, u_2}^{Y_1} = \rho_{x_1, u_2}^{Y_1}$ i.e.
$\rho_{q, u_1, x_1, u_2}^{Y_1}$ is independent of $q$ and $u_1$.

Let $t := 2^{R'_2}$.
For $\tilde{m}_2 \in [t]$, define the cq-state
\begin{eqnarray*}
\lefteqn{\rho^{Q U_1 X_1 (U_2)^{t} Y_1}(\tilde{m}_2) } \\
& := &
\sum_{q, u_1, x_1, u_2^{t}} 
p(q) p(u_1, x_1 | q) p(u_2^{t} | q) \\
&  &
~~~~~~~~~~
\ketbra{q, u_1, x_1, u_2^{t}}^{Q U_1 X_1 (U_2)^{t}}
\otimes \rho_{q, u_1, x_1, u_2^t(\tilde{m}_2)}^{Y_1}.
\end{eqnarray*}
These states will play the role of 
$\rho'(1), \ldots, \rho'(t)$ in
Fact~\ref{fact:gencqtypical}.

Define the cq-states
\[
\rho^{Q U_1 X_1 U_2 Y_1}_{
(\{Q U_1 U_2\}, \{X_1\}, \{\})
},
\rho^{U_0 U_1 X_1 U_2 Y_1}_{
(\{Q U_2\}, \{U_1 X_1\}, \{\})
},
\rho^{U_0 U_1 X_1 U_2 Y_1}_{
(\{Q U_1\}, \{U_2 X_1\}, \{\})
},
\rho^{U_0 U_1 X_1 U_2 Y_1}_{
(\{Q\}, \{U_2 U_1 X_1\}, \{\})
},
\]
as in Claim~4 of Fact~\ref{fact:cqtypical}.
We can now define the cq-state
\begin{eqnarray*}
\lefteqn{
\rho^{Q U_1 X_1 (U_2)^{t} Y_1}_{
(\{Q U_1 U_2\}, \{X_1\}, \{\})
}(\tilde{m}_2)
} \\
& := &
\sum_{q, u_1, u_2^{t}} 
p(q) p(u_1 | q) p(u_2^{t} | q) 
\ketbra{q, u_1, u_2^{t}}^{Q U_1 (U_2)^{t}} \\
&   &
~~~~~~~~~~~~
{}  \otimes 
(\sum_{x_1} p(x_1 | q u_1) \ketbra{x_1}^{X_1}) \otimes
\rho_{q, u_1, u_2^t(\tilde{m}_2)}^{Y_1}.
\end{eqnarray*}
In other words, 
$
\rho^{Q U_1 X_1 (U_2)^{t} Y_1}_{
(\{Q U_1 U_2\}, \{X_1\}, \{\})
}(\tilde{m}_2)
$
is the cq-state obtained 
by `naturally extending' $U_2$ to $(U_2)^{t}$ and `embedding'
$
\rho^{Q U_1 X_1 U_2 Y_1}_{
(\{Q U_1 U_2\}, \{X_1\}, \{\})
}
$
at `$\tilde{m}_2$th position'. Similarly, we can define the quantum
state
$
\rho^{Q U_1 X_1 (U_2)^t Y_1}_{
(\{Q U_2\}, \{U_1 X_1\}, \{\})
}(\tilde{m}_2).
$

Fix $0 < \alpha, \delta < 1$. 
Fact~\ref{fact:gencqtypical} tells us that there is an augmentation of
the classical systems
$Q$, $U_1$, $X_1$, $U_2$ to 
$\hQ := Q \otimes \cL$, 
$\hU_1 := U_1 \otimes \cL$,
$\hX_1 := X_1 \otimes \cL$,
$\hU_2 := U_2 \otimes \cL$,
an extension $\hY_1$ of the quantum system $Y_1$ i.e.
$Y_1 \otimes \C^2 \otimes \C^2 \otimes \C^{t+1} \leq \hY_1$,  cq-states
$(\rho')^{\hQ \hU_1 \hX_1 (\hU_2)^{t} \hY_1}(\tilde{m}_2)$ and 
a cq-POVM element 
$(\hPi)^{\hQ \hU_1 \hX_1 (\hU_2)^{t} \hY_1}$
such that
\begin{enumerate}

\item
\begin{eqnarray*}
\lefteqn{(\rho')^{\hQ \hU_1 \hX_1 (\hU_2)^{t} \hY_1}(\tilde{m}_2)} \\
& =  &
|\cL|^{-(t+3)} \cdot {} \\
&    &
~~~
\sum_{q, u_1, x_1, u_2^{t}, \vecl} 
p(q) p(u_1, x_1 | q) p(u_2^{t} | q) \\
&   &
~~~~~~~~~~~~~~~~~~~~
\ketbra{q, u_1, x_1, u_2^{t}}^{Q U_1 X_1 (U_2)^{t}} \\
&   &
~~~~~~~~~~~~~~~~~~~~~~~~~
{} \otimes
\ketbra{\vecl}^{\cL^{\otimes (t+3)}} \otimes
(\rho')_{q, u_1, x_1, u_2^t(\tilde{m}_2),\vecl^{\cL^{\otimes 4}}}^{Y'_1} \\
&   &
~~~~~~~~~~~~~~~~~~~~~~~~~
{} \otimes
\ketbra{0}^{\C^2} \otimes \ketbra{0}^{\C^{t+1}},
\end{eqnarray*}
for some quantum states 
$(\rho')_{q, u_1, x_1, u_2, \vecl^{\cL^{\otimes 4}}}^{Y'_1}$;

\item
\begin{eqnarray*}
\lefteqn{(\hPi)^{\hQ \hU_1 \hX_1 (\hU_2)^{t} \hY_1}} \\
& =  &
\sum_{q, u_1, x_1, u_2^{t}, \vecl} 
\ketbra{q, u_1, x_1, u_2^{t}}^{Q U_1 X_1 (U_2)^{t}} \\
&    &
~~~~~~~~~~~~~~~~
{} \otimes
\ketbra{\vecl}^{\cL^{\otimes (t+3)}} \otimes
(\hPi)_{q, u_1, x_1, u_2^t, \vecl}^{\hY_1},
\end{eqnarray*}
for some POVM elements
$(\hPi)_{q, u_1, x_1, (u_2)^t, \vecl}^{\hY_1}$;

\item
\[
\begin{array}{l}
\left\|
(\rho')^{\hQ \hU_1 \hX_1 (\hU_2)^t \hY_1}(\tilde{m}_2) 
\right. \\
~~~~~~
{} -
\rho^{Q U_1 X_1 (U_2)^t Y_1}(\tilde{m}_2) \otimes 
\ketbra{0}^{\C^2} \otimes
\frac{\one^{\cL^{\otimes (t+3)}}}{|\cL|^{t+3}} \\
~~~~~~~~~~~~~~
\left.
{} \otimes
\ketbra{0}^{\C^2} \otimes
\ketbra{0}^{\C^{t+1}}
\right\|_1 \\
\; \leq \;
2^4 \delta;
\end{array}
\]

\item
\[
\begin{array}{l}
\Tr [(\hPi)^{\hQ \hU_1 \hX_1 (\hU_2)^t \hY_1} 
(\rho')^{\hQ \hU_1 \hX_1 (\hU_2)^t \hY_1}(\tilde{m}_2) 
] \\
\; \geq \;
1 -
2^{2^9} \cdot 3^4 \cdot \delta^{-2} \epsilon - 
2^4 \delta - \alpha;
\end{array}
\]

\item
Define 
\[
(\rho')^{\hQ \hU_1 \hX_1 (\hU_2)^{t} \hY_1}_{
(\{\hQ \hU_1 \hU_2\}, \{\hX_1\}, \{\})
}(\tilde{m}_2),
(\rho')^{\hQ \hU_1 \hX_1 (\hU_2)^t \hY_1}_{
(\{\hQ \hU_2\}, \{\hU_1 \hX_1\}, \{\})
}(\tilde{m}_2)
\]
analogously as the corresponding quantities 
\[
\rho^{Q U_1 X_1 (U_2)^{t} Y_1}_{
(\{Q U_1 U_2\}, \{X_1\}, \{\})
}(\tilde{m}_2),
\rho^{Q U_1 X_1 (U_2)^t Y_1}_{
(\{Q U_2\}, \{U_1 X_1\}, \{\})
}(\tilde{m}_2)
\]
defined above.
Then,
\begin{eqnarray*}
\lefteqn{
\Tr [(\hPi)^{\hQ \hU_1 \hX_1 (\hU_2)^t \hY_1} 
(\rho')^{\hQ \hU_1 \hX_1 (\hU_2)^t \hY_1}_{
(\{\hQ \hU_1 \hU_2\}, \{\hX_1\}, \{\})
}(\tilde{m}_2)
] 
} \\
& \leq &
\frac{1-\alpha}{\alpha}
\left(
\sum_{j \neq \tilde{m}_2}
2^{-D^\epsilon_H(
\rho^{Q U_1 X_1 (U_2)^t Y_1}(j) \|
\rho^{Q U_1 X_1 (U_2)^{t} Y_1}_{
(\{Q U_1 U_2\}, \{X_1\}, \{\})
}(\tilde{m}_2)
)
}
\right. \\
&     &
~~~~~~~~~~~~
\left.
{} +
2^{-D^\epsilon_H(
\rho^{Q U_1 X_1 (U_2)^t Y_1}(\tilde{m}_2) \|
\rho^{Q U_1 X_1 (U_2)^{t} Y_1}_{
(\{Q U_1 U_2\}, \{X_1\}, \{\})
}(\tilde{m}_2)
)
}
\right) \\
& \leq &
\frac{1-\alpha}{\alpha}
(
2^{R'_2} 
2^{-I^\epsilon_H(X_1 U_2 : Y_1 | Q U_1)}
+ 2^{-I^\epsilon_H(X_1 : Y_1 | Q U_1 U_2)}
),
\end{eqnarray*}
where the hypothesis mutual information quantities are computed with
respect to the cq-state $\rho^{Q U_1 X_1 U_2 X_2 Y_1}$.
The last inequality follows because the POVM element optimising
the hypothesis testing mutual information quantitiy
$I^\epsilon_H(X_1 U_2 : Y_1 | Q U_1)$, when applied at the `$j$th 
position', $j \neq m_2$, accepts $\rho^{Q U_1 X_1 (U_2)^t Y_1}(j)$ 
with probability at least $1 - \epsilon$ and accepts 
$
\rho^{Q U_1 X_1 (U_2)^{t} Y_1}_{
(\{Q U_1 U_2\}, \{X_1\}, \{\})
}(\tilde{m}_2)
$
with probability at most 
$2^{-I^\epsilon_H(X_1 U_2 : Y_1 | Q U_1)}$. 
Similarly, the POVM element optimising
the hypothesis testing mutual information quantitiy
$I^\epsilon_H(X_1 : Y_1 | Q U_1 U_2)$, when applied at the 
`$\tilde{m}_2$th position', accepts 
$\rho^{Q U_1 X_1 (U_2)^t Y_1}(\tilde{m}_2)$ 
with probability at least $1 - \epsilon$ and accepts 
$
\rho^{Q U_1 X_1 (U_2)^{t} Y_1}_{
(\{Q U_1 U_2\}, \{X_1\}, \{\})
}(\tilde{m}_2)
$
with probability at most 
$2^{-I^\epsilon_H(X_1 : Y_1 | Q U_1 U_2)}$. 
Arguing along the same lines, we get
\begin{eqnarray*}
\lefteqn{
\Tr [(\hPi)^{\hQ \hU_1 \hX_1 (\hU_2)^t \hY_1} 
(\rho')^{\hQ \hU_1 \hX_1 (\hU_2)^t \hY_1}_{
(\{\hQ \hU_2\}, \{\hU_1 \hX_1\}, \{\})
}(\tilde{m}_2)
] 
} \\
& \leq &
\frac{1-\alpha}{\alpha}
(
2^{R'_2} 
2^{-I^\epsilon_H(X_1 U_1 U_2 : Y_1 | Q)}
+ 2^{-I^\epsilon_H(U_1 X_1 : Y_1 | Q U_2)}
). 
\end{eqnarray*}
\end{enumerate}

For $(m''_1, m'_1) \in [2^{R_1 - R'_1}] \times [2^{R'_1}]$, define 
the POVM element $\Lambda_{m''_1, m'_1}^{Y_1}$ 
as follows: Attach an ancilla of 
$\ketbra{0}^{\C^2} \otimes \ketbra{0}^{\C^2} \otimes \ketbra{0}^{\C^{t+1}}$
to register $Y_1$ and then apply the POVM element 
$
\Lambda_{
(u_1, l'_1)(m'_1), (x_1, l''_1)(m'_1, m''_1)
}^{\hY_1}.
$
Here 
$
\Lambda_{
(u_1, l'_1)(m'_1), (x_1, l''_1)(m'_1, m''_1)
}^{\hY_1}.
$
is a POVM element from the PGM
constructed, for the augmented codebook $\cC'$, from the set of 
positive operators
\[
\begin{array}{l}
\left\{
(\hPi)_{
(q, l_0), (u_1, l'_1)(\tilde{m}'_1), 
(x_1, l''_1)(\tilde{m}'_1, \tilde{m}''_1), 
\{(u_2, l'_2)(\tilde{m}'_2): \tilde{m}'_2 \in [2^{R'_2}]\},
\delta
}^{\hY_1}: 
\right. \\
~~~~~~
\left.
(\tilde{m}'_1, \tilde{m}''_1) \in [2^{R_1 - R'_1}] \times [2^{R'_1}]
\right\},
\end{array}
\]
which in turn is provided by Fact~\ref{fact:gencqtypical}.
Observe that $\Lambda_{m'_1, m''_1}^{Y_1}$ only depends on the entries
$q$, $(u_1, l'_1)(\tilde{m}'_1)$, 
$(x_1, l''_1)(\tilde{m}'_1, \tilde{m}''_1)$,
$(u_2, l'_2)(\tilde{m}'_2)$, 
$\tilde{m}'_1 \in [2^{R'_1}]$,
$\tilde{m}''_1 \in [2^{R_1 - R'_1}]$,
$\tilde{m}'_2 \in [2^{R'_2}]$ of $\cC'$.
Similarly for $(m''_2, m'_2) \in [2^{R_2 - R'_2}] \times [2^{R'_2}]$,
we can define the POVM element 
$\Lambda_{m'_2, m''_2}^{Y_2}$.
Receiver $B_1$ applies his POVM to the contents of $Y_1$ 
and outputs the result 
$\hat{m}_1 := (\hat{m}'_1, \hat{m}''_1)$ as his guess 
for $m_1 = (m'_1, m''_1)$.
Similarly, receiver $B_2$ applies his POVM to the contents of $Y_2$ 
and outputs the result 
$\hat{m}_2 := (\hat{m}'_2, \hat{m}''_2)$ as his guess 
for $m_2 = (m'_2, m''_2)$.

\medskip

\noindent
{\bf Error probability:}

\noindent
Suppose the senders $A_1$, $A_2$ transmit $(m_1, m_2)$. 
For a state 
$
\sigma^{X'_1}_{x_1(m'_1, m''_1)} \otimes  
\sigma^{X'_2}_{x_2(m'_2, m''_2)}
$
inputted into the channel $\chan$, denoted its output state at
$B_1$ by 
$
\rho^{Y_1}_{x_1(m'_1, m''_1), x_2(m'_2, m''_2)}.
$
We bound the
expected decoding error of $B_1$ over the choice of a random 
augmented codebook $\cC'$ as follows: 
\begin{eqnarray*}
\lefteqn{
\E_{\cC'}[\Pr[\mbox{$B_1$'s error}]]
} \\
&   =  &
\E_{\cC'}[
\Tr [
(\one^{Y_1} - \Lambda^{Y_1}_{m'_1, m''_1})
\rho_{x_1(m'_1, m''_1), x_2(m'_2, m''_2)}^{Y_1}
]
] \\
&   =  &
\E_{\cC'}[
\Tr [
(\one^{Y_1} - \Lambda^{Y_1}_{m'_1, m''_1})
\rho_{q, u_1(m'_1), x_1(m'_1, m''_1), u_2(m'_2), x_2(m'_2, m''_2)}^{Y_1}
]
] \\
&   =  &
\E_{\cC'}[
\Tr [
(\one^{Y_1} - \Lambda^{Y_1}_{m'_1, m''_1})
\rho_{q, u_1(m'_1), x_1(m'_1, m''_1), u_2(m'_2)}^{Y_1}
]
] \\
&   =  &
\E_{\cC'}[
\Tr [
(
\one^{\hY_1} - 
\Lambda^{\hY_1}_{(u_1, l'_1)(m'_1), (x_1, l''_1)(m'_1, m''_1)}
) \\
&      &
~~~~~~~~~~~~~~~
(
\rho_{q, u_1(m'_1), x_1(m'_1, m''_1), u_2(m'_2)}^{Y_1} \otimes
\ketbra{0}^{\C^2} \\
&      &
~~~~~~~~~~~~~~~~~~~~~~
{} \otimes \ketbra{0}^{\C^2} \otimes \ketbra{0}^{\C^{t+1}}
)
]
] \\
& \leq &
\E_{\cC'}[
\Tr [
(
\one^{\hY_1} - 
\Lambda^{\hY_1}_{(u_1, l'_1)(m'_1), (x_1, l''_1)(m'_1, m''_1)}
) \\
&      &
~~~~~~~~~~~~~~~~~
(
(\rho')_{(q,l_0), (u_1, l'_1)(m'_1), (x_1, l''_1)(m'_1, m''_1), 
         (u_2, l'_2)(m'_2)}^{Y'_1} \\
&      &
~~~~~~~~~~~~~~~~~~~~~~
{} \otimes
\ketbra{0}^{\C^2} \otimes \ketbra{0}^{\C^{t+1}}
)
]
] \\
&      &
{} +
\E_{\cC'}
\left[
\frac{1}{2}
\left\|
(\rho')_{(q,l_0), (u_1, l'_1)(m'_1), (x_1, l''_1)(m'_1, m''_1), 
         (u_2, l'_2)(m'_2)}^{Y'_1} 
\right.
\right. \\
&      &
~~~~~~~~~~~~~~~~
{} \otimes
\ketbra{0}^{\C^2} \otimes \ketbra{0}^{\C^{t+1}} \\
&     &
~~~~~~~~~~~~~~~~~~~~~~
{} -
\rho_{q, u_1(m'_1), x_1(m'_1, m''_1), u_2(m'_2)}^{Y_1} \otimes
\ketbra{0}^{\C^2} \\
&     &
~~~~~~~~~~~~~~~~~~~~~~~~~~~~~
\left.
\left.
{} \otimes \ketbra{0}^{\C^2} \otimes \ketbra{0}^{\C^{t+1}}
\right\|_1
\right] \\
& \stackrel{a}{\leq} &
2 \E_{\cC'}[
\Tr [
(
\one^{\hY_1} - {} \\
&      &
~~~~~~~~~~~~~~~~~
(\hPi)^{\hY_1}_{
(q, l_0), (u_1, l'_1)(m'_1), (x_1, l''_1)(m'_1, m''_1),
\{(u_2, l'_2)(\tilde{m}_2): \tilde{m}_2 \in [2^{R'_2}]\},
\delta
}
) \\
&   &
~~~~~~~~~~~~~~~~~~~
(
(\rho')_{(q,l_0), (u_1, l'_1)(m'_1), (x_1, l''_1)(m'_1, m''_1), 
         (u_2, l'_2)(m'_2)}^{Y'_1} \\
&   &
~~~~~~~~~~~~~~~~~~~~~~~~
{} \otimes
\ketbra{0}^{\C^2} \otimes \ketbra{0}^{\C^{t+1}}
)
]
] \\
&      &
{} +
4 \sum_{\tilde{m}''_1 \neq m''_1}
\E_{\cC'}[ \\
&      &
~~~~~~~~~~~~~~
\Tr [
(\hPi)^{\hY_1}_{
(q, l_0), (u_1, l'_1)(m'_1), (x_1, l''_1)(m'_1, \tilde{m}''_1),
\{(u_2, l'_2)(\tilde{m}_2): \tilde{m}_2 \in [2^{R'_2}]\},
\delta
} \\
&     &
~~~~~~~~~~~~~~~~~~~~~~~~~~~
(
(\rho')_{(q,l_0), (u_1, l'_1)(m'_1), (x_1, l''_1)(m'_1, m''_1), 
         (u_2, l'_2)(m'_2)}^{Y'_1} \\
&     &
~~~~~~~~~~~~~~~~~~~~~~~~~~~~~~~~~
{} \otimes
\ketbra{0}^{\C^2} \otimes \ketbra{0}^{\C^{t+1}}
)
]
] \\
&      &
{} +
4 \sum_{\tilde{m}'_1 \neq m'_1, \tilde{m}''_2}
\E_{\cC'}[
\Tr [
(\hPi)^{\hY_1}_{
(q, l_0), (u_1, l'_1)(\tilde{m}'_1), 
(x_1, l''_1)(\tilde{m}'_1, \tilde{m}''_1),
\{(u_2, l'_2)(\tilde{m}_2): \tilde{m}_2 \in [2^{R'_2}]\},
\delta
} \\
&     &
~~~~~~~~~~~~~~~~~~~~~~~~~~~~~~~~~~~~~
(
(\rho')_{(q,l_0), (u_1, l'_1)(m'_1), (x_1, l''_1)(m'_1, m''_1), 
         (u_2, l'_2)(m'_2)}^{Y'_1} \\
&     &
~~~~~~~~~~~~~~~~~~~~~~~~~~~~~~~~~~~~~~~~~
{} \otimes
\ketbra{0}^{\C^2} \otimes \ketbra{0}^{\C^{t+1}}
)
]
] \\
&      &
{} +
|\cL|^{-4} \cdot {} \\
&      &
~~~~~
\sum_{q,u_1,x_1,u_2,l_0,l'_1,l''_1,l'_2}
p(q) p(u_1,x_1 | q) p(u_2 | q) \\
&      &
~~~~~~~~~~~~~~~~~~
\frac{1}{2}
\left\|
(\rho')_{(q,l_0), (u_1, l'_1), (x_1, l''_1), 
         (u_2, l'_2)}^{Y'_1} \otimes
\ketbra{0}^{\C^2} \otimes \ketbra{0}^{\C^{t+1}}
\right. \\
&      &
~~~~~~~~~~~~~~~~~~~~~~~~~
\left.
{} -
\rho_{q, u_1, x_1, u_2}^{Y_1} \otimes
\ketbra{0}^{\C^2} \otimes \ketbra{0}^{\C^2} \otimes \ketbra{0}^{\C^{t+1}}
\right\|_1 \\
&   =  &
2 |\cL|^{-(t+3)} \cdot {} \\
&      &
~~~~
\sum_{q,u_1,x_1,u_2^t,l_0,l'_1,l''_1,(l'_2)^t}
p(q) p(u_1,x_1 | q) p(u_2^t | q) \\
&   &
~~~~~~~~~~~~~~~~~~~~~~~~
\Tr [
(
\one^{\hY_1} - 
(\hPi)^{\hY_1}_{
(q, l_0), (u_1, l'_1), (x_1, l''_1), (u_2^t, (l'_2)^t), \delta
}
) \\
&   &
~~~~~~~~~~~~~~~~~~~~~~~~~~~~~~~~~~
(
(\rho')_{(q,l_0), (u_1, l'_1), (x_1, l''_1),
         (u_2^t, (l'_2)^t)(m'_2)}^{Y'_1} \\
&   &
~~~~~~~~~~~~~~~~~~~~~~~~~~~~~~~~~~~~~~~
{} \otimes
\ketbra{0}^{\C^2} \otimes \ketbra{0}^{\C^{t+1}}
)
] \\
&      &
{} +
4 |\cL|^{-(t+4)} \cdot {} \\
&      &
~~~
\sum_{\tilde{m}''_1 \neq m''_1}
\sum_{q,u_1,x_1,\tilde{x}_1,u_2^t,l_0,l'_1,l''_1,\tilde{l}''_1,(l'_2)^t} \\
&   &
~~~~~~~~~
p(q) p(u_1 | q) p(x_1 | u_1 q) p(\tilde{x}_1 | u_1 q) p(u_2^t | q) \\
&   &
~~~~~~~~~~~~~
\Tr [
(\hPi)^{\hY_1}_{
(q, l_0), (u_1, l'_1), (\tilde{x}_1, \tilde{l}''_1), u_2^t, (l'_2)^t,
\delta
} \\
&     &
~~~~~~~~~~~~~~~~~~~~~~
(
(\rho')_{(q,l_0), (u_1, l'_1), (x_1, l''_1),
         (u_2^t, (l'_2)^t)(m'_2)}^{Y'_1} \\
&     &
~~~~~~~~~~~~~~~~~~~~~~~~~~
{} \otimes
\ketbra{0}^{\C^2} \otimes \ketbra{0}^{\C^{t+1}}
)
] \\
&      &
{} +
4 |\cL|^{-(t+5)} \cdot {} \\
&      &
~~~
\sum_{\tilde{m}'_1 \neq m'_1, \tilde{m}''_2}
\sum_{q,u_1,x_1,\tilde{u}_1,\tilde{x}_1, u_2^t,
      l_0,l'_1,l''_1,\tilde{l}'_1,\tilde{l}''_1,(l'_2)^t} \\
&   &
~~~~~~~~~~~~
p(q) p(u_1, x_1 | q) p(\tilde{u}_1, \tilde{x}_1 | q) p(u_2^t | q) \\
&   &
~~~~~~~~~~~~~~~~~
\Tr [
(\hPi)^{\hY_1}_{
(q, l_0), (\tilde{u}_1, \tilde{l}'_1), (\tilde{x}_1, \tilde{l}''_1),
u_2^t, (l'_2)^t, \delta
} \\
&     &
~~~~~~~~~~~~~~~~~~~~~~~~~~~
(
(\rho')_{(q,l_0), (u_1, l'_1), (x_1, l''_1),
         (u_2^t, (l'_2)^t)(m'_2)}^{Y'_1} \\
&     &
~~~~~~~~~~~~~~~~~~~~~~~~~~~~~~~~
{} \otimes
\ketbra{0}^{\C^2} \otimes \ketbra{0}^{\C^{t+1}}
)
] \\
&      &
{} +
\frac{1}{2}
\left\|
(\rho')^{\hQ \hU_1 \hX_1 (\hU_2)^t \hY_1}(m'_2) - {} 
\right. \\
&      &
~~~~~~~~~~~~~~~
\rho^{Q U_1 X_1 (U_2)^t Y_1}(m'_2) \otimes
\ketbra{0}^{\C^2} \otimes 
\frac{\one^{\cL^{\otimes (t+3)}}}{|\cL|^{t+3}} \\
&      &
~~~~~~~~~~~~~~~~~
\left.
{} \otimes
\ketbra{0}^{\C^2} \otimes \ketbra{0}^{\C^{t+1}}
\right\|_1 \\
&   =  &
2 \Tr [
(
\one^{\hQ \hU_1 \hX_1 (\hU_2)^t \hY_1} 
- (\hPi)^{\hQ \hU_1 \hX_1 (\hU_2)^t \hY_1}
)
(\rho')^{\hQ \hU_1 \hX_1 (\hU_2)^t \hY_1}(m'_2)
] \\
&  &
{} +
4 \cdot (2^{R_1 - R'_1} - 1) \cdot {} \\
&  &
~~~~~~~~~~~~~
\Tr [
(\hPi)^{\hQ \hU_1 \hX_1 (\hU_2)^t \hY_1}
(\rho')_{(\{\hQ \hU_1 \hU_2\}, \{\hX_1\}, \{\} )}^{
\hQ \hU_1 \hX_1 (\hU_2)^t \hY_1
}(m'_2)
] \\
&  &
{} +
4 \cdot (2^{R'_1} - 1) 2^{R_1 - R'_1} \cdot {} \\
&  &
~~~~~~~~~~~~~
\Tr [
(\hPi)^{\hQ \hU_1 \hX_1 (\hU_2)^t \hY_1}
(\rho')_{(\{\hQ \hU_2\}, \{\hU_1 \hX_1\}, \{\} )}^{
\hQ \hU_1 \hX_1 (\hU_2)^t \hY_1
}(m'_2)
] \\ 
&      &
{} +
\frac{1}{2}
\left\|
(\rho')^{\hQ \hU_1 \hX_1 (\hU_2)^t \hY_1}(m'_2) - {} 
\right. \\
&      &
~~~~~~~~~~~~~~
\rho^{Q U_1 X_1 (U_2)^t Y_1}(m'_2) \otimes
\ketbra{0}^{\C^2} \\
&      &
~~~~~~~~~~~~~~~~~~~
\left.
{} \otimes 
\frac{\one^{\cL^{\otimes (t+3)}}}{|\cL|^{t+3}} \otimes
\ketbra{0}^{\C^2} \otimes \ketbra{0}^{\C^{t+1}}
\right\|_1 \\
& \leq &
(2^{2^{12}} \delta^{-2} \epsilon + 2^5 \delta + 2 \alpha) \\
&    &
{} +
\frac{1-\alpha}{\alpha} 
2^{R_1 - R'_1 + 2} (
2^{R'_2 - I^\epsilon_H(X_1 U_2 : Y_1 | Q U_1)} + 
2^{- I^\epsilon_H(X_1 : Y_1 | Q U_1 U_2)} 
) \\ 
&    &
{} +
\frac{1-\alpha}{\alpha} 
2^{R_1 + 2} (
2^{R'_2 - I^\epsilon_H(X_1 U_2 : Y_1 | Q)} + 
2^{- I^\epsilon_H(X_1 : Y_1 | Q U_2)} 
) \\
& \leq &
(2^{2^{12}} \delta^{-2} \epsilon + 2^5 \delta + 2 \alpha) \\
&    &
{} +
\frac{1-\alpha}{\alpha} 
(
3^{R_1 - R'_1 + R'_2 + 2 - I^\epsilon_H(X_1 U_2 : Y_1 | Q U_1)} \\
&    &
~~~~~~~~~~~~~~~~~~
{} + 
2^{R_1 - R'_1 + 2 - I^\epsilon_H(X_1 : Y_1 | Q U_1 U_2)} 
) \\ 
&    &
{} +
\frac{1-\alpha}{\alpha} 
(
2^{R_1 + 2 + R'_2 - I^\epsilon_H(X_1 U_2 : Y_1 | Q)} + 
2^{R_1 + 2- I^\epsilon_H(X_1 : Y_1 | Q U_2)} 
),
\end{eqnarray*}
where Step~(a) follows from Fact~\ref{fact:HN}.
Setting $\delta := \epsilon^{1/3}$, $\alpha := \epsilon^{2/3}$ we get that 
$
\E_{\cC'}[\Pr[\mbox{$B_1$'s error}]] \leq
2^{2^{13}} \epsilon^{1/3}.
$

\begin{figure*}
\begin{center}
\includegraphics[width=\textwidth]{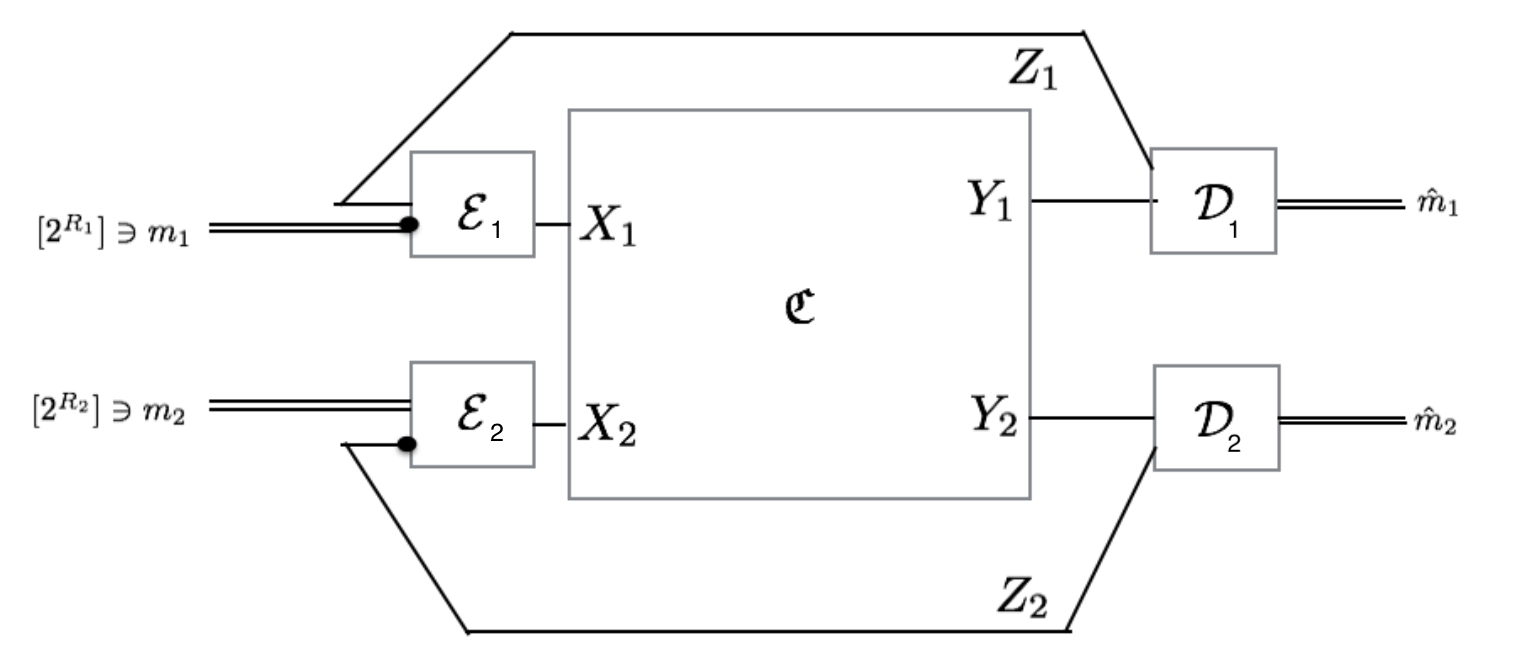}
\end{center}
\caption{Quantum interference channel with cis entanglement assistance.}
\label{fig:interferenceassisted}
\end{figure*}
Similarly, 
$
\E_{\cC'}[\Pr[\mbox{$B_2$'s error}]] \leq
2^{2^{13}} \epsilon^{1/3}.
$
Thus, there is an augmented codebook $\cC'$ such that sum of 
average decoding errors of $B_1$ and $B_2$ is at most 
$2^{2^{14}} \epsilon^{1/3}$.
The average probability that at least one of $B_1$ or $B_2$ err for
$\cC'$ is thus seen 
to be at most $2^{2^{14}} \epsilon^{1/6}$ using Fact~\ref{fact:gentle}.
This finishes the proof of one-shot Chong-Motani-Garg-El Gamal inner 
bound.
\end{proof}

It is possible to give a one-shot Han-Kobayashi style inner bound for
the interference channel also. To do so, we need to use the one-shot
simultaneous decoder for the three sender multiple access channel
constructed in \cite{sen:oneshot}. In contrast to the iid setting,
in the one-shot setting
it is not known if the Han-Kobayashi and Chong-Motani-Garg-El Gamal
rate regions are the same or not. This is because  we do not 
have good chain rules for the hypothesis testing mutual information. 

An advantage of the Han-Kobayashi inner bound technique is that it 
can be easily extended to give a non-trivial inner bound for the
interference channel wih entanglement assistance 
(see Figure~\ref{fig:interferenceassisted}). The 
Chong-Motani-Garg-El Gamal inner bound technique does not seem to be
suitable for this endeavour. We consider the case of an interference
channel with independent prior entanglement between $A_1$ and $B_1$,
and between $A_2$ and $B_2$, which seems
to be the most natural scenario. We shall call this the interference
channel with {\em cis entanglement}. For this channel, we can obtain
the following inner bound.
\begin{theorem}[One-shot Han-Kobayashi, ent. assist.]
Let $\chan: X_1 X_2 \rightarrow Y_1 Y_2$ be a quantum interference 
channel. We are allowed use of arbitrary amount of prior entanglment 
between $A_1$ and $B_1$, and $A_2$ and $B_2$. 
Let $\cQ$, $\cU_1$, $\cU_2$ be
three new sample spaces and $(Q, U_1, U_2)$ be a 
jointly distributed random variable with probability mass function
$p(q) p(u_1 | q) p(u_2 | q)$. 
Let $\cX''_1$, $\cZ_1$, $\cX''_2$, $\cZ_2$ be four new Hilbert spaces.
Let $\psi_1^{X'_1 Z_1} \otimes \psi_2^{X'_2 Z_2}$ be a tensor
product quantum state in $X'_1 Z_1 X'_2 Z_2$.
For every element $u_1 \in \cU_1$, let 
$\cE_{1, u_1}^{X'_1 \rightarrow X_1}$
be a fixed encoding superoperator; similarly for every
$u_2 \in \cU_2$.
Consider the classical quantum state
\begin{eqnarray*}
\lefteqn{\rho^{Q U_1 U_2 Y_1 Z_1 Y_2 Z_2} } \\
& := &
\sum_{q, u_1, u_2}
p(q) p(u_1 | q) p(u_2 | q)
\ketbra{q, u_1, u_2}^{Q U_1 U_2} \\
&   &
~~~~~~~~~~~~~
{} \otimes
((\chan^{X_1 X_2 \rightarrow Y_1 Y_2} \otimes \I^{Z_1 Z_2})( \\
&   &
~~~~~~~~~~~~~~~~~~~~~~~
((\cE_{1,u_1}^{X'_1 \rightarrow X_1} \otimes \I^{Z_1})(
\psi_1^{X'_1 Z_1}))^{X_1 Z_1} \otimes {} \\
&   &
~~~~~~~~~~~~~~~~~~~~~~~~~~~
((\cE_{2,u_2}^{X'_2 \rightarrow X_2} \otimes \I^{Z_2})(
\psi_1^{X'_2 Z_2}))^{X_2 Z_2}
)
)^{Y_1 Z_1 Y_2 Z_2}.
\end{eqnarray*}
Let $R'_1$, $R''_1$, $R'_2$, $R''_2$, $\epsilon$, 
be such that
\begin{eqnarray*}
R'_1 
& \leq &
I_H^{\epsilon}(U_1 : Y_1 U_2 Z_1 | Q) - 2 -  \log \frac{1}{\epsilon} \\
R''_1 
& \leq &
I_H^{\epsilon}(Z_1 : Y_1 U_1 U_2 | Q) - 2 -  \log \frac{1}{\epsilon} \\
R'_1 + R''_1 
& \leq &
I_H^{\epsilon}(U_1 Z_1  : Y_1 U_2 | Q) - 2 -  \log \frac{1}{\epsilon} \\
R'_1 + R'_2 
& \leq &
I_H^{\epsilon}(U_1 U_2 : Y_1 Z_1 | Q) - 2 -  \log \frac{1}{\epsilon} \\
R''_1 + R'_2 
& \leq &
I_H^{\epsilon}(Z_1 U_2 : Y_1 U_1 | Q) - 2 -  \log \frac{1}{\epsilon} \\
R'_1 + R''_1 + R'_2 
& \leq &
I_H^{\epsilon}(U_1 U_2 Z_1 : Y_1 | Q) - 2 -  \log \frac{1}{\epsilon} \\
& & \\
R'_2 
& \leq &
I_H^{\epsilon}(U_2 : Y_2 U_1 Z_2 | Q) - 2 -  \log \frac{1}{\epsilon} \\
R''_2 
& \leq &
I_H^{\epsilon}(Z_2 : Y_2 U_1 U_2 | Q) - 2 -  \log \frac{1}{\epsilon} \\
R'_2 + R''_2 
& \leq &
I_H^{\epsilon}(U_2 Z_2  : Y_2 U_1 | Q) - 2 -  \log \frac{1}{\epsilon} \\
R'_1 + R'_2 
& \leq &
I_H^{\epsilon}(U_1 U_2 : Y_2 Z_2 | Q) - 2 -  \log \frac{1}{\epsilon} \\
R''_2 + R'_1 
& \leq &
I_H^{\epsilon}(Z_2 U_1 : Y_2 U_2 | Q) - 2 -  \log \frac{1}{\epsilon} \\
R'_2 + R''_2 + R'_1 
& \leq &
I_H^{\epsilon}(U_1 U_2 Z_2 : Y_2 | Q) - 2 -  \log \frac{1}{\epsilon} \\
\end{eqnarray*}
where the mutual
information quantitites above are computed with respect to the
cq-state $\rho^{Q U_1 U_2 Y_1 Z_1 Y_2 Z_2}$. 
Define $R_1 := R'_1 + R''_1$, $R_2 := R'_2 + R''_2$.
Then there exists an 
$(R_1, R_2, 2^{2^{14}} \epsilon^{1/6})$-quantum interference channel code 
for sending
classical information through $\chan$ with cis entanglement assistance.
\end{theorem}
\begin{proof}
We employ rate splitting as in the traditional Han-Kobayashi inner bound
proof. We split the message $m_1$ into a common message $m'_1$ and a 
personal message $m''_1$; similarly for $m_2$. The message triple
$(m''_1, m'_1, m'_2)$ is sent to $B_1$ by treating the interference
channel as a three sender multiple access channel. Message $m''_1$
is transmitted using the position based coding technique which
requires the assistance of many independent copies of the 
state $\psi_1^{X'_1 Z}$ with the $X'_1$ parts under the possession
of $A_1$ and the $Z_1$ parts under the possession of $B_1$. Message
$m'_1$ is transmitted without entanglement assistance by applying
the encoding superoperator $\cE_{1, u_1(m'_1)}$ to 
register $X'_1(u_1(m'_1))$, and feeding its output
to the input register $X_1$  of the channel $\chan$.
Similar statements hold for messages $m''_2$ and $m'_2$.
An analogous consideration holds for receiver $B_2$.

We obtain the above region by employing simultaneous decoders for
the two three-sender multiple access channels with receivers $B_1$ and 
$B_2$ induced by $\chan$. The simultaneous decoders have to handle
both entanglement assisted as well as unassisted messages, so in a
sense, they are the hybrid of the two simultaneous decoders described
in \cite{sen:oneshot}. Another important difference from the standard 
decoders
for the multiple access channel is that we do not want the additional
constraints
\begin{eqnarray*}
R'_2 
& \leq &
I_H^{\epsilon}(U_2 : Y_1 U_1 Z_1 | Q) - 2 -  \log \frac{1}{\epsilon}, \\
R'_1 
& \leq &
I_H^{\epsilon}(U_1 : Y_2 U_2 Z_2 | Q) - 2 -  \log \frac{1}{\epsilon}
\end{eqnarray*}
to appear in the rate region. For this, we need to use `union of 
intersection of POVM elements', which can be done by appealing to
Fact~\ref{fact:gencqtypical}. The `union' expresses the observation
that it is unnecessary for $B_1$ to decode $m'_2$ if he has already
successfully decoded $(m'_1, m''_1)$, or equivalently, decoding
$m'_2$ wrongly is not a problem if $B_1$ has already successfully
decoded $(m'_1, m''_1)$.  A similar comment holds for
$B_2$. In the asymptotic iid setting, presence
of these constraints does not affect the rate region. In the one-shot
setting it is not clear if this is true, simply because of the lack
of chain rules for the hypothesis testing mutual information. In
the interest of obtaining as large an inner bound as possible, we
use the `union' technique.
\end{proof}

\medskip

\noindent
{\bf Remark:}

\noindent
It is possible to get another Han-Kobayashi style inner bound if we
allow independent prior entanglement between all the four possible 
sender-receiver pairs i.e. if we allow both cis and trans entanglement,
where trans entanglement refers to prior entanglement between 
$A_1$ and $B_2$, and $A_2$ and $B_1$.
In this scenario, we can directly employ, as a subroutine, the 
entanglement assisted
one-shot inner bound for the three sender quantum multiple access channel
described in \cite{sen:oneshot}. All the message parts will now
be transmitted using entanglement assistance. However, we do not discuss
this further in this paper because we feel that trans entanglement is
an unnatural resource.

\section{Conclusions}
\label{sec:conclusions}
In this paper, we have fruitfully used the quantum joint typicality 
lemmas from \cite{sen:oneshot} to prove some novel inner bounds for 
sending classical information through multiterminal quantum channels.
All our inner bounds require us to construct simultaneous decoders,
and hold in the one-shot setting. 
For some of these problems, one-shot inner bounds were hitherto unknown
even in the classical setting. All our one-shot inner bounds are strong 
enough to reduce to the standard inner bounds in the asymptotic iid
limit, and provide non-trivial second order rates.

The Han-Kobayashi inner bound for a quantum interference channel with 
entanglement
assistance given in this paper does not use the prior entanglement
to send the common parts of the messages.
It seems that there is scope for improvement in this regard, which is 
left for future work.

It will be interesting to find other applications of simultaneous
decoders in quantum network information theory. Already, Ding,
Gharibyan, Hayden and Walter~\cite{ding:relay} have 
used the joint typicality
lemmas to construct a simultaneous decoder for a particular quantum
relay channel.

The quantum joint typicality lemmas give us robust tools to handle
union and intersection for `packing type' problems.
However, they fail for `covering type' problems. Covering type problems
often arise in source coding. Constructing simultaneous decoders for
them remains a major open problem.

\section*{Acknowledgements}
I thank Patrick Hayden, David Ding and
Hrant Gharibyan for useful discussions, and Mark Wilde for pointers
to important references. I am grateful to the anonymous referees
of an earlier version of the paper,
whose comments helped greatly in improving the presentation.

\balance

\bibliography{simultaneous}

\end{document}